\documentclass[a4paper,aps,pra,showpacs,superscriptaddress,reprint,nofootinbib,floatfix]{revtex4-1}
\usepackage[utf8]{inputenc}
\usepackage[T1]{fontenc}
\usepackage[USenglish]{babel}
\usepackage[sc,osf]{mathpazo}
\usepackage[scaled=0.86]{berasans}
\usepackage[scaled=1.03]{inconsolata}
\usepackage[colorlinks]{hyperref}
\usepackage{graphicx}
\usepackage{amsmath,amssymb,amsthm}
\usepackage{xspace}
\usepackage{mathtools}
\usepackage{amsmath}
\usepackage{lipsum}
\usepackage{dsfont}

\usepackage[usenames,dvipsnames]{xcolor}
             \usepackage{siunitx}
             \usepackage{pgf,tikz}\usetikzlibrary{intersections, decorations.pathreplacing, calc, decorations.markings}
             \usepackage{pgfplots}

\usepackage{paralist}

\usepackage{enumitem}
\usepackage[normalem]{ulem}
\usepackage{units}

\DeclareMathOperator{\tr}{tr}

\DeclareMathOperator{\conv}{conv}

\let\originalleft\left
\let\originalright\right
\renewcommand{\left}{\mathopen{}\mathclose\bgroup\originalleft}
\renewcommand{\right}{\aftergroup\egroup\originalright}

\newcommand{\KetBra}[1]{{\Ket{#1}\!\Bra{#1} }}

\newcommand{\Bra}[1]{{ \langle \! \langle{#1}\vert }}
\newcommand{\Ket}[1]{{ \vert {#1}  \rangle \!  \rangle}}

\newcommand{\bra}[1]{\left\langle #1 \right|}
\newcommand{\ket}[1]{\left| #1 \right\rangle}

\newcommand{\ketbra}[2]{\left|#1\middle\rangle\middle\langle#2\right|}
\newcommand{\proj}[1]{\left|#1\middle\rangle\middle\langle#1\right|}

\newcommand{\prin}[2]{\left\langle#1, #2\right\rangle}
\newcommand{\norm}[1]{\left\|#1\right\|}

\newcommand{\midset}{\ |\ }
\newcommand{\wsep}{W^{\text{sep}}}
\newcommand{\wsepcone}{{\mathcal W}^{\text{sep}}}

\newcommand{\com}{{[\, ,\,]}}
\newcommand{\ant}{{\{\, ,\,\}}}

\newcommand{\de}[1]{\left(#1\right)}
\newcommand{\De}[1]{\left[#1\right]}
\newcommand{\DE}[1]{\left\{#1\right\}}

\newcommand{\mathand}{\quad\text{and}\quad}

\newcommand{\Hi}{\mathcal{H}}

\newcommand{\id}{\mathds{1}}

\newcommand{\rr}{R_{\text{r}}}
\newcommand{\gr}{R_{\text{g}}}

\newcommand{\ie}{i.e.\@\xspace}

\newtheorem{theorem}{Theorem}
\newtheorem*{theorem*}{Theorem}
\newtheorem{lemma}[theorem]{Lemma}
\newtheorem{definition}[theorem]{Definition}
\newtheorem{corollary}[theorem]{Corollary}

\mathtoolsset{centercolon}

\renewcommand{\ketbra}[2]{\left|#1\middle\rangle\!\middle\langle#2\right|}
\renewcommand{\proj}[1]{\left|#1\middle\rangle\!\middle\langle#1\right|}

\hypersetup{pdftitle={{Witnessing causal nonseparability}}}

\begin{document}
\title{Witnessing causal nonseparability}
\author{Mateus Araújo}
\affiliation{Faculty of Physics, University of Vienna, Boltzmanngasse 5 1090 Vienna, Austria}
\affiliation{Institute for Quantum Optics and Quantum Information (IQOQI), Boltzmanngasse 3 1090 Vienna, Austria}
\author{Cyril Branciard}
\affiliation{Institut Néel, CNRS and Universit\'e Grenoble Alpes, 38042 Grenoble Cedex 9, France}
\author{Fabio Costa}
\affiliation{Faculty of Physics, University of Vienna, Boltzmanngasse 5 1090 Vienna, Austria}
\affiliation{Institute for Quantum Optics and Quantum Information (IQOQI), Boltzmanngasse 3 1090 Vienna, Austria}
\affiliation{Centre for Engineered Quantum Systems, School of Mathematics and Physics, The University of Queensland, St Lucia, QLD 4072, Australia}
\author{Adrien Feix}
\affiliation{Faculty of Physics, University of Vienna, Boltzmanngasse 5 1090 Vienna, Austria}
\affiliation{Institute for Quantum Optics and Quantum Information (IQOQI), Boltzmanngasse 3 1090 Vienna, Austria}
\author{Christina Giarmatzi}
\affiliation{Centre for Engineered Quantum Systems, School of Mathematics and Physics, The University of Queensland, St Lucia, QLD 4072, Australia}
\affiliation{Centre for Quantum Computer and Communication Technology, School of Mathematics and Physics, University of Queensland, Brisbane,   QLD 4072, Australia}
\author{Časlav Brukner}
\affiliation{Faculty of Physics, University of Vienna, Boltzmanngasse 5 1090 Vienna, Austria}
\affiliation{Institute for Quantum Optics and Quantum Information (IQOQI), Boltzmanngasse 3 1090 Vienna, Austria}
\date{\today}

\begin{abstract}
 Our common understanding of the physical world deeply relies on the notion that events are ordered with respect to some time parameter, with past events serving as causes for future ones. Nonetheless, it was recently found that it is possible to formulate quantum mechanics without any reference to a global time or causal structure. The resulting framework includes new kinds of quantum resources that allow performing tasks -- in particular, the violation of \textit{causal inequalities} -- which are impossible for events ordered according to a global causal order. However, no physical implementation of such resources is known. Here we show that a recently demonstrated resource for quantum computation -- the \textit{quantum switch} -- is a genuine example of ``indefinite causal order''. We do this by introducing a new tool -- the \textit{causal witness} -- which can detect the \textit{causal nonseparability} of any quantum resource that is incompatible with a definite causal order. We show however that the quantum switch does not violate any causal inequality. 
\end{abstract}
\maketitle

\section{Introduction}

It is commonly assumed that information is processed through a series of operations which are performed according to a specific order. This is justified by the assumption of a global, underlying time parameter according to which all operations can be ordered. A convenient representation of this structure is that of a circuit~\cite{deutsch1989quantum}, Fig.~\ref{circuit}(a), in which systems are ``wires'' that connect ``boxes'', which represent operations performed on the systems. At a more abstract level, a circuit only imposes a given \textit{causal structure} between operations, as the time order between operations that can be performed in parallel is irrelevant.
The circuit framework is also ubiquitous in the study of quantum foundations to formalize generalized, possibly post-quantum, probabilistic theories~\cite{hardy2009foliable, coecke2010quantum, PhysRevA.81.062348, PhysRevA.84.012311}.

It has been suggested that such a framework might be too restrictive to encompass the most general kinds of information processing allowed by quantum physics~\cite{chiribella09}. 
For example, one can consider protocols in which the order between different operations is controlled by a quantum degree of freedom. It has been shown that such protocols exploiting a so-called ``quantum switch'' not only provide computational advantage over standard, time-ordered, ones~\cite{chiribella12, araujo14}, but they are also physically realizable and a first experimental proof-of-principle has been recently demonstrated~\cite{procopio_experimental_2014}. At a more fundamental level, an underlying time or causal order might not be well-defined in a theory that combines the dynamical causal structure of general relativity and the probabilistic nature of quantum mechanics \cite{hardy2007towards,rovelli1990quantum,ashtekar96}.
 
It is therefore natural to ask what the most general resources allowed by quantum mechanics beyond the circuit model are. In Ref.~\cite{oreshkov12} the \textit{process matrix formalism} was proposed as a general framework to describe resources that can be accessed in ``local laboratories'' and which are locally in agreement with quantum physics, Fig.~\ref{circuit}(b).

\begin{figure}[htpc]
  \centering
  \begin{tikzpicture}[scale=1.5]
                \draw[thick, red] (0.9,-0.7) arc (125:55:-1.56) (-0.9,-0.7) -- (0.9,-0.7);
                \draw[thick, red] (-0.6,-0.7) -- (-0.6,2.3) (0,-0.7) --(0,2.3) (0.6,-0.7) -- (0.6,2.3);
		\node[draw, thick, rectangle,minimum width=1.6cm,minimum height=0.4cm, fill=white!80!gray] (A) at (0.3,-0.4) {$\mathcal{M}_{1}$};
		\node[draw, thick, rectangle,minimum width=1.6cm,minimum height=0.4cm, fill=white!80!gray] (B) at (0.3,0.6) {$\mathcal{M}_{2}$};
		\node[draw, thick, rectangle,minimum width=0.4cm,minimum height=0.4cm, fill=white!80!gray] (C) at (-0.6,1.1) {$\mathcal{M}_{3}$};
		\node[draw, thick, rectangle,minimum width=0.4cm,minimum height=0.4cm, fill=white!80!gray] (D) at (0,2.1) {$\mathcal{M}_4$};
		\node[draw, thick, red, rectangle,minimum width=2.5cm,minimum height=0.5cm, fill=white] (U) at (0,0.1) {};
		\node[draw, thick, red, rectangle,minimum width=0.75cm,minimum height=0.5cm, fill=white] (U2) at (0.6,1.1) {};
		\node[draw, thick, red, rectangle,minimum width=1.6cm,minimum height=0.5cm, fill=white] (U3) at (-0.3,1.6) {};
                \node[] (label) at (0,-1.4) {(a)};
                \node[] (hidden) at (2.0,0) {};
  \end{tikzpicture}
  \begin{tikzpicture}[scale=1.5]
		\node[draw, thick, rectangle,minimum width=0.4cm,minimum height=0.4cm,fill=white!80!gray] (A) at (0,-0.4) {$\mathcal{M}_{1}$};
		\node[draw, thick, rectangle,minimum width=0.4cm,minimum height=0.4cm,fill=white!80!gray] (B) at (0.6,0.6) {$\mathcal{M}_{2}$};
		\node[draw, thick, rectangle,minimum width=0.4cm,minimum height=0.4cm,fill=white!80!gray] (C) at (0,0.6) {$\mathcal{M}_{3}$};
		\node[draw, thick, rectangle,minimum width=0.4cm,minimum height=0.4cm,fill=white!80!gray] (D) at (0.6,1.6) {$\mathcal{M}_{4}$};
                \draw[thick, red]  (-0.7,-1.1) -- (-0.7, 2.3) --(0.9,2.3) -- (0.9,1.9) -- (-0.3,1.9) -- (-0.3,1.3) -- (0.9,1.3) -- (0.9,0.9) -- (-0.3,0.9) -- (-0.3,0.3) -- (0.9,0.3) -- (0.9,-0.1) -- (-0.3, -0.1) -- (-0.3,-0.7) -- (0.9,-0.7) -- (0.9,-1.1) -- (-0.7,-1.1);
                \draw[thick, red,->] (A) -- +(0,0.3);  \draw[thick, red,<-]  (A) -- +(0,-0.3);
                \draw[thick, red,->] (B) -- +(0,0.3);  \draw[thick, red,<-]  (B) -- +(0,-0.3);
                \draw[thick, red,->] (C) -- +(0,0.3);  \draw[thick, red,<-]  (C) -- +(0,-0.3);
                \draw[thick, red,->] (D) -- +(0,0.3);  \draw[thick, red,<-]  (D) -- +(0,-0.3);
                \node[] (label) at (0.2,-1.4) {(b)};
  \end{tikzpicture}
  \caption{(a) If the operations $\mathcal{M}_{i}$ of local agents are performed in a definite causal sequence, they can be represented as gates in a circuit, where information flows from bottom to top. (b) A process matrix formalizes a resource in which the order between operations may not be fixed. A probabilistic mixture of different orders is an example of a process matrix that does not correspond to a circuit. Still, in this case operations are performed in a well-defined order in each experimental run; the most general resource with this property is called \textit{causally separable}. The process matrix formalism also allows for the more general case of \textit{causally nonseparable} resources~\cite{oreshkov12}.}
  \label{circuit}
\end{figure}
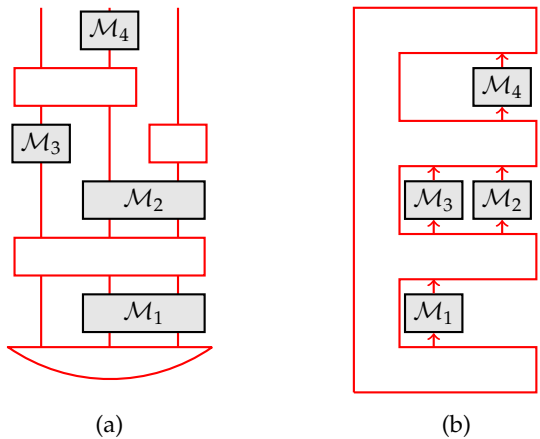

Causal relations are defined operationally in this formalism. If, for example, through appropriate state preparations, an agent $A$ can influence the outcomes of measurements performed by an agent $B$, whereas $B$ is never able to influence $A$, then $A$ causally precedes $B$ by definition and, in this case, the physical resources available to them can in fact be represented as a circuit. A first example of a resource that cannot be represented as a circuit is a \textit{probabilistic mixture} of circuits: a definite order still exists between $A$ and $B$ in each run of an experiment, but which order is realized in a given run is only specified according to some probability distribution. Resources compatible with a definite causal order, in this broader sense, are called \textit{causally separable}. Surprisingly, the formalism also allows for \textit{causally nonseparable} resources, which are incompatible with any definite order between operations. It was found that a set of agents with access to a specific causally nonseparable resource could perform a task, the violation of a \textit{causal inequality}, which is impossible for arbitrary causally ordered strategies, even allowing probabilistic mixtures of orders~\cite{oreshkov12}. However, there is no physical interpretation for such resources and no physically realizable protocol is known which can violate a causal inequality. 

It is therefore not completely clear what is the precise relation between ``quantum correlations with no causal order'', which violate causal inequalities, and physically implementable resources, such as the quantum switch, which outperform causally ordered ones. To understand this relation, a crucial observation is that the causal inequalities are \textit{device-independent} constraints: they are formulated independently of the physics of the systems or the specific apparatuses employed. On the other hand, the tasks discussed in Refs.~\cite{chiribella12, araujo14} include additional assumptions, as for example that in each laboratory quantum systems of a definite dimension have to be used. It is clear that, given additional restrictions, it is more difficult for causally-ordered agents to perform certain tasks and, consequently, it can be easier to detect the lack of causal order in a physical resource.

The aim of the present work is to develop a general framework for the device-\textit{dependent} detection of causal nonseparability.
The central tool we introduce is what we call a \textit{causal witness}, which represents a set of quantum operations, such as unitaries, channels, state preparations, and measurements, whose expectation value is non-negative as long as all the operations are performed in a definite causal order, \ie, as long as only causally separable resources are used. The observation of a negative expectation value is thus sufficient to conclude that the operations were not performed in a definite order. The concept is analogous to that of entanglement witness: an observable that has a non-negative expectation value for separable states but can have a negative expectation value for specific entangled states.

We find that, for every causally nonseparable process, it is possible to construct a causal witness that detects it. Importantly, and differently from the case of entanglement witnesses, it is possible to use this method to write necessary and sufficient conditions for causal separability in a form that can be checked efficiently using semidefinite programming (SDP).

The tools developed are applied to the study of the quantum switch as a resource within the process matrix formalism.
We show that, indeed, the quantum switch corresponds to a causally nonseparable process. We show that the protocol of Ref.~\cite{chiribella12} can be reformulated as a causal witness which detects the causal nonseparability of the quantum switch. We also find new, more efficient witnesses, which could be useful for experimental implementations.

We finally address the question of whether the quantum switch can pass any device-independent test of causal nonseparability. As it turns out, this is not possible: we prove that a broad class of resources, including the quantum switch, cannot violate any causal inequality.

The paper is organized as follows: In Section~\ref{sec:formalism}, we review the process matrix formalism, giving a convenient characterization of general and causally separable process matrices for the cases of interest. In Section~\ref{sec:causal_witness}, we introduce and characterize the central concept of causal witness, and we present efficient algorithms for finding witnesses and for proving the causal (non)separability of a general process matrix. In Section \ref{sec:Q_control} we formalise the quantum switch as a process matrix. We proceed to prove its causal nonseparability in Section~\ref{sec:switch_witness}, through the use of causal witnesses. One such witness is the task proposed in Ref.~\cite{chiribella12}, that we optimize to increase its resistance to noise. Finally, we clarify in Section~\ref{sec:causal_ineqs} the link between causal witnesses and causal inequalities and show that the quantum switch cannot violate any causal inequality.
\medskip

\section{The process matrix formalism}\label{sec:formalism}

In the general scenario we consider in this paper, $N$ parties $A^i$ establish correlations by exchanging physical systems between their laboratories. Each party opens their laboratory only once to let an incoming system enter and to send an outgoing system out; they can act on these systems by performing an arbitrary operation in their local laboratory, which can yield different measurement outcomes. The causal relations between the parties (\ie, the ordering of events) are not \emph{a priori} specified. The most general situation compatible with the assumption that \emph{the operations performed in each local laboratory can be described by the quantum formalism} can be conveniently represented in the ``process matrix'' formalism introduced in Ref.~\cite{oreshkov12}. This extends the ``comb'' formalism of Ref.~\cite{chiribella09b}, which describes causally ordered quantum networks. The aim of the formalism is to characterize all possible probability distributions that can be obtained in our general scenario. The key concept is that of a \textit{process}, which can be understood as the external resource determining the statistics of the local operations, and which generalizes both the notions of quantum state and of quantum channel. The \textit{process matrix} is a useful mathematical representation of such a concept. We shall use these two terms interchangeably.

\subsection{Local operations}
Each party $A$ acts in a \textit{local quantum laboratory}, which can be identified by an input Hilbert space ${\cal H}^{A_I}$ and an output Hilbert space ${\cal H}^{A_O}$. The dimensions $d_{A_I}$ and $d_{A_O}$ of input and output spaces do not have to be equal, as ancillary systems can be added or discarded during an operation; we shall nevertheless assume throughout the paper that all Hilbert spaces are finite-dimensional. 
According to quantum theory, the most general local operation is described by a completely positive (CP), trace non-increasing map ${\cal M}^A: A_I\rightarrow A_O$ \cite{chuang00}, where we write $A_I$, respectively $A_O$, for the space of hermitian linear operators over the Hilbert space ${\cal H}^{A_I}$, resp. ${\cal H}^{A_O}$. Examples of CP maps are deterministic operations, such as unitaries or quantum channels, or (generalized) measurements. In general, a label $a$, denoting the measurement outcome, is associated with the CP map ${\cal M}_a^A$. The choice of operation (e.g. of measurement setting) is represented by an \textit{instrument}~\cite{davies70}, which is defined as the collection ${\cal J}^A = \left\{{\cal M}_a^A\right\}_{a=1}^m$ of CP maps associated to all measurement outcomes, characterized by the property that $\sum_{a=1}^m{\cal M}_a^A$ is CP and trace-preserving (CPTP). An instrument generalizes the notion of POVM (positive operator-valued measure) to include the transformations applied to the system; it reduces to a POVM for $1$-dimensional output spaces. When the choice of operation is described by a classical variable $x$, we will express such a dependence explicitly as ${\cal J}^A_x =\left\{{\cal M}_{a|x}^A\right\}_{a=1}^m$.

A convenient representation of CP maps is given by the Choi-Jamio{\l}kowski (CJ) isomorphism \cite{zyczkowski06}. For a CP map ${\cal M}_a^A : A_I\rightarrow A_O$, its corresponding CJ matrix is defined here as
\begin{equation}
M^{A_I A_O}_a:=\left[{\cal I}\otimes{\cal M}^A_a\left(\KetBra{\id}\right)\right]^T \in A_I\otimes A_O,
\label{CJ}
\end{equation}
where ${\cal I}$ is the identity map, $\Ket{\id} \equiv \Ket{\id}^{A_I A_I}:=\sum_j \ket{j}^{A_I}\otimes\ket{j}^{A_I} \in \Hi^{A_I} \otimes \Hi^{A_I}$ is a (non-normalized) maximally entangled state, and $T$ denotes matrix transposition with respect to the chosen orthonormal basis $\{\ket{j}^{A_I}\}$ of $\Hi^{A_I}$. Some useful properties of the CJ isomorphism are given in Appendix~\ref{sec:cj}. A map is completely positive if and only if its CJ representation is positive semidefinite, while the trace-preserving condition is equivalent to $\tr_{A_O} M^{A_I A_O} = \id^{A_I}$ (where $\tr_{A_O}$ denotes the partial trace over $A_O$, and $\id^{A_I}$ is the identity matrix in $A_I$). An instrument is therefore equivalently represented as a set
\begin{equation}
\left\{M^{A_I A_O}_a\right\}_{a=1}^m, \quad
 M^{A_I A_O}_a\geq 0, \quad
 \tr_{A_O}\sum_{a=1}^m M^{A_I A_O}_a = \id^{A_I}.
\label{instrument}
\end{equation}

\subsection{Process matrices}\label{sec:w_matrix}

As discussed in Ref.~\cite{oreshkov12}, requiring that quantum mechanics holds locally implies that the probability that the $N$ parties $A^i$ observe the outcomes $a_1,\ldots,a_N$, for a choice of operations $x_1,\ldots,x_N$, is a multilinear function $P\big({\cal M}_{a_1|x_1}^{A^1},\ldots,{\cal M}_{a_N|x_N}^{A^N}\big)$ of the corresponding CP maps ${\cal M}_{a_1|x_1}^{A^1},\ldots,{\cal M}_{a_N|x_N}^{A^N}$. Using the CJ representation, it was shown that these probabilities can then be expressed as
\begin{eqnarray}
&& P\big(M_{a_1|x_1}^{A^1},\ldots,M_{a_N|x_N}^{A^N}\big) \nonumber \\
&& \qquad = \ \tr \left[ \Big(M^{A^1_I A^1_O}_{a_1|x_1}\otimes \ldots \otimes M^{A^N_I A^N_O}_{a_N|x_N} \Big) W \right],
\label{born}
\end{eqnarray}
for some hermitian operator $W\in A^1_I\otimes A^1_O\otimes \ldots \otimes A^N_I\otimes A^N_O $ called a \textit{process matrix}, which describes the general quantum resource connecting the local laboratories.

The set of valid process matrices is defined by requiring that probabilities are well-defined -- that is, they must be non-negative and must sum up to $1$ -- for all possible operations, including operations that involve, in each laboratory, local interactions with ancillary systems that may be entangled with the other laboratories.
As we show in Appendix~\ref{sec:valid_w}, these conditions are equivalent to
\begin{gather}
 \label{W_pos}
 W \ge 0 , \\
\label{normalization}
 \tr W = d_O, \\
 \label{LV}
 W = L_V(W),
\end{gather}
where $d_O = d_{A^1_O}\ldots d_{A^N_O}$, and $L_V$ is a projector onto the linear subspace ${\cal L}_V \subset A^1_I\otimes A^1_O\otimes \ldots \otimes A^N_I\otimes A^N_O$ defined in Appendix~\ref{sec:valid_w}.
We will denote the closed convex cone of non-normalized processes defined by \eqref{W_pos} and~\eqref{LV} by $\mathcal{W}$.

In the case of two parties $A$ (Alice) and $B$ (Bob), see Figure~\ref{fig:w-bipartite}, these conditions on $W \in A_I \otimes A_O \otimes B_I \otimes B_O$ reduce to 
\begin{gather}
 W \ge 0 \, , \\
 \tr W = d_O \, , \\
 {}_{B_IB_O}W = {}_{A_OB_IB_O}W \, , \label{eq:valid_bipartite_1} \\
 {}_{A_IA_O}W = {}_{A_IA_OB_O}W \, , \label{eq:valid_bipartite_2} \\
 W = {}_{B_O}W + {}_{A_O}W - {}_{A_OB_O}W \, , \label{eq:valid_bipartite_3}
\end{gather}
where (here and throughout the paper) the operator ${}_{X}\cdot$ denotes the CPTP map consisting in tracing out the subsystem $X$ and replacing it by the normalized identity operator, formally defined as
\begin{equation}\label{def:notation_trace}
 _X W = \frac{\id^{X}}{d_X} \otimes \tr_X W \, .
\end{equation}

\begin{figure}[htpc]
  \centering
  \begin{tikzpicture}[scale=1.6]
		\node[draw, thick, rectangle,minimum width=0.7cm,minimum height=0.9cm, fill=white!80!gray] (S) at (-0.8,0) {$A$};
		\node[red,thick, minimum width=0.7cm,minimum height=0.9cm] (W) at (0,0) {$W$};
		\node[draw, thick, rectangle,minimum width=0.7cm,minimum height=0.9cm, fill=white!80!gray] (T) at (0.8,0) {$B$};
                \draw[red, thick] (-1.2,-0.85) -- (1.2,-0.85) -- (1.2,-0.5) -- (0.35,-0.5) -- (0.35,0.5) -- (1.2,0.5) -- (1.2,0.85) -- (-1.2,0.85) -- (-1.2, 0.5) -- (-0.35,0.5) -- (-0.35,-0.5) -- (-1.2,-0.5) -- (-1.2,-0.85);
                \draw[->, red, thick] (S) -- (-0.8,0.5);
                \draw[<-, red, thick] (S) -- (-0.8,-0.5);
                \draw[->, red, thick] (T) -- (0.8,0.5);
                \draw[<-, red, thick] (T) -- (0.8,-0.5);
		\node[red] (SO) at ([shift={(-0.15cm,0.4cm)}]S) {\footnotesize $A_O$};
		\node[red] (SI) at ([shift={(-0.15cm,-0.4cm)}]S) {\footnotesize $A_I$};
		\node[red] (TO) at ([shift={(0.15cm,0.4cm)}]T) {\footnotesize $B_O$};
		\node[red] (TI) at ([shift={(0.15cm,-0.4cm)}]T) {\footnotesize $B_I$};
  \end{tikzpicture}
  \caption{Representation of a bipartite process matrix $W$, connecting Alice's ($A_O$) and Bob's ($B_O$) output systems to their input systems ($A_I$ and $B_I$).}
  \label{fig:w-bipartite}
\end{figure}
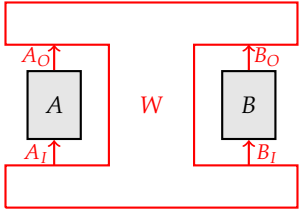

\subsubsection{Non-signalling and 1-way-signalling process matrices}
\label{subsubsec:no_sig_1_way_sig}

Two important particular cases of process matrices may shed light on the above definition. The first case is when the process matrix does not allow for any signalling, and the second one is when it allows for signalling only in one fixed direction between the parties. They are discussed in more details in Appendix~\ref{sec:appendix_process_matrices}.

The first case is described by process matrices $W$ satisfying
\begin{equation} \label{eq:no_signaling}
 W = {}_{A^1_O\ldots A^N_O}W = \rho^{A^1_I\ldots A^N_I}\otimes \id^{A^1_O\ldots A^N_O},
\end{equation}
where $\rho^{A^1_I\ldots A^N_I}$ is a density matrix representing an ordinary quantum state. In this case, the probability rule~\eqref{born} reduces to the standard Born rule
\begin{equation}
P\big(M_{a_1|x_1}^{A^1},\ldots,M_{a_N|x_N}^{A^N}\big) = \tr \left[ \Big(E^{A^1_I}_{a_1|x_1}\otimes \ldots \otimes E^{A^N_I}_{a_N|x_N} \Big) \rho \right],
\label{trueborn}
\end{equation}
where $E^{A^i_I}_{a_i|x_i} := \tr_{A^i_O} M^{A^i_I A^i_O}_{a_i|x_i}$ are POVM elements.

\medskip

The second case, of which the first one is a particular case, is described by process matrices $W$ satisfying
\begin{eqnarray} \label{eq:causal_order}
\begin{array}{rcl}
W &=& {}_{A^N_O}W, \\
{}_{A^N_I A^N_O}W &=& {}_{A^{N-1}_O A^N_I A^N_O}W, \\[-1mm]
&\vdots& \\
{}_{A^2_I A^2_O \ldots\ A^N_I A^N_O}W &=& {}_{A^1_O A^2_I A^2_O \ldots\ A^N_I A^N_O}W \, .
\end{array}
\end{eqnarray}
These conditions, first found in \cite{gutoski06,chiribella09b}, mean that party $A^i$ can only signal to party $A^j$ if $i < j$. The process is therefore compatible with the causal order $A^1 \prec A^2 \prec\, \ldots \prec A^N$. When this is the case, we write as a mnemonic
\begin{equation}
 W = W^{A^1 \prec A^2 \prec\, \ldots \prec A^N}.
\end{equation}
Process matrices of this form (and the obvious permutations) are called \textit{causally ordered}. As shown in Refs.~\cite{gutoski06,chiribella09b}, they correspond to standard (causally ordered) quantum circuits and can be implemented as quantum channels with memory between the parties.

\subsubsection{Bipartite causally separable processes}\label{sec:causally_separable_2}

According to Eq.~\eqref{eq:causal_order}, a bipartite causally ordered process matrix $W^{A \prec B}$ compatible with the order $A \prec B$ satisfies $W^{A \prec B} = {}_{B_O}W^{A \prec B}$ and ${}_{B_I B_O}W^{A \prec B} = {}_{A_O B_I B_O}W^{A \prec B}$.
Note that the latter relation corresponds to Eq.~\eqref{eq:valid_bipartite_1} above, and is therefore automatically satisfied if $W^{A \prec B} \in \mathcal L_V$. Thus, a given matrix $W^{A \prec B} \in A_I \otimes A_O \otimes B_I \otimes B_O$ is a valid causally ordered process matrix compatible with the order $A \prec B$ if and only if
$W^{A \prec B} \ge 0$, $\tr W^{A \prec B} = d_O$,
\begin{equation}
W^{A \prec B} \in \mathcal L_V \mathand W^{A \prec B} = {}_{B_O}W^{A \prec B} \, . \label{eq:valid_W_AB}
\end{equation}
The analogous condition holds for the order $B \prec A$.

Note that a non-signalling process matrix $W$ must be compatible with both orders $A \prec B$ and $B \prec A$. It must therefore satisfy $W = {}_{B_O}W = {}_{A_O}W$, or equivalently $W = {}_{A_OB_O}W$; we indeed recover the form of Eq.~\eqref{eq:no_signaling}.

Following Ref.~\cite{oreshkov12}, we say that a bipartite process matrix $W$ is \textit{causally separable} if it can be decomposed as a convex combination of causally ordered processes, \ie, if it is of the form
\begin{equation}
 W^{\text{sep}} \, = \, q \, W^{A \prec B} \, + \, (1{-}q) \, W^{B \prec A}, \label{def:caus_sep}
\end{equation}
with $0 \leq q \leq 1$.
Ignoring the normalization constraint, the set of causally separable process matrices is a convex cone, which we denote by $\wsepcone$.
A process matrix that cannot be decomposed as in~\eqref{def:caus_sep} is called \textit{causally nonseparable}.

\subsubsection{Tripartite causally separable processes}\label{sec:causally_separable_3c}

In this paper we will define tripartite causal separability only for processes where the output space of the third party $C$ (Charlie) is trivial, \ie, $d_{C_{O}} = 1$ (see Figure~\ref{fig:w-tripartite}). As $C$ cannot signal to the other parties, every process of this kind if compatible with $C$  being last. Thus, only two causal orders are relevant in this case: $A \prec B \prec C$ and $B \prec A \prec C$. The conditions for process matrices being compatible with these orders are, according to equation \eqref{eq:causal_order},
\begin{align}
 W^{A \prec B \prec C} &= {}_{C_O} W^{A \prec B \prec C}, \\
 {}_{C_IC_O}W^{A \prec B \prec C} &= {}_{B_OC_IC_O} W^{A \prec B \prec C}, \\
 {}_{B_IB_OC_IC_O}W^{A \prec B \prec C} &= {}_{A_OB_IB_OC_IC_O} W^{A \prec B \prec C},
\end{align}
and
\begin{align}
 W^{B \prec A \prec C} &= {}_{C_O} W^{B \prec A \prec C}, \\
 {}_{C_IC_O}W^{B \prec A \prec C} &= {}_{A_OC_IC_O} W^{B \prec A \prec C}, \\
 {}_{A_IA_OC_IC_O}W^{B \prec A \prec C} &= {}_{B_OA_IA_OC_IC_O} W^{B \prec A \prec C}.
\end{align}
Since these three conditions together define a linear subspace, we can write them more succinctly as
\begin{align}
 W^{A \prec B \prec C} &= L_{A \prec B \prec C}(W^{A \prec B \prec C}), \label{eq:labc} \\
 W^{B \prec A \prec C} &= L_{B \prec A \prec C}(W^{B \prec A \prec C}), \label{eq:lbac}
\end{align}
where $L_{A \prec B \prec C}$ and $L_{B \prec A \prec C}$ are the projectors onto the aforementioned subspaces.

\begin{figure}[htpc]
  \centering
  \begin{tikzpicture}[scale=1.6]
		\node[draw, thick, rectangle,minimum width=0.7cm,minimum height=0.85cm, fill=white!80!gray] (S) at (-0.8,0) {$A$};
		\node[red, minimum width=0.7cm,minimum height=0.85cm] (W) at (0,0) {$W$};
		\node[draw, thick, rectangle,minimum width=0.7cm,minimum height=0.85cm, fill=white!80!gray] (T) at (0.8,0) {$B$};
                \node[draw, thick, rectangle,minimum width=0.7cm,minimum height=0.85cm, fill=white!80!gray] (C) at (0,1.35) {$C$};
                 \draw[red, thick] (-1.2,-0.85) -- (1.2,-0.85) -- (1.2,-0.5) -- (0.35,-0.5) -- (0.35,0.5) -- (1.2,0.5) -- (1.2,0.85) -- (-1.2,0.85) -- (-1.2, 0.5) -- (-0.35,0.5) -- (-0.35,-0.5) -- (-1.2,-0.5) -- (-1.2,-0.85);
                \draw[thick, red, ->] (S) -- (-0.8,0.5);
                \draw[thick, red, <-] (S) -- (-0.8,-0.5);
                \draw[thick, red, ->] (T) -- (0.8,0.5);
                \draw[thick, red, <-] (T) -- (0.8,-0.5);
                \draw[thick, red, <-] (C) -- (0,0.85);
		\node[red] (SO) at ([shift={(-0.15cm,0.38cm)}]S) {\footnotesize $A_O$};
		\node[red] (SI) at ([shift={(-0.15cm,-0.4cm)}]S) {\footnotesize $A_I$};
		\node[red] (TO) at ([shift={(0.15cm,0.38cm)}]T) {\footnotesize $B_O$};
		\node[red] (TI) at ([shift={(0.15cm,-0.4cm)}]T) {\footnotesize $B_I$};
                \node[red] (CI) at ([shift={(0.15cm,-0.4cm)}]C) {\footnotesize $C_I$};
  \end{tikzpicture}
  \caption{Representation of a tripartite process matrix $W$ where one party has trivial output $d_{C_{O}} = 1$. It can be seen as connecting Alice's ($A_O$) and Bob's ($B_O$) output systems to Alice, Bob and Charlie's input systems $A_I$, $B_I$ and $C_{I}$.}
  \label{fig:w-tripartite}
\end{figure}
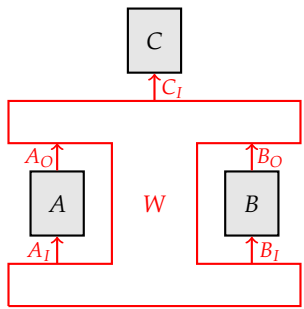

Therefore, when $C$'s output space is trivial, we will call a tripartite process matrix $W^{\text{sep}}$ causally separable if it is of the form
\begin{equation}\label{eq:wsep_switch}
W^{\text{sep}} \, = \, q \, W^{A\prec B\prec C} \, + \, (1{-}q) \, W^{B\prec A\prec C},
\end{equation}
with $0 \leq q \leq 1$. Ignoring the normalization constraint, this defines a convex cone $\wsepcone_{3C}$.
We will use this definition in Section~\ref{sec:switch_witness} to show that a recently introduced tripartite quantum resource, which yields information-processing advantages with respect to causally ordered processes~\cite{chiribella12,araujo14}, is causally nonseparable.

The generalization of the notion of causal separability to a larger number of parties, with arbitrary dimensions of the output spaces, is not trivial. The reason is that one can consider situations in which an agent, through her local operations, could modify a classical variable that determines the causal order of agents in her future. In such a ``classical switch'', operations would still be causally ordered in each run of an experiment, but it wouldn't be possible to write the corresponding process matrix as a mixture of causally ordered ones. As this issue does not affect the cases treated here, we shall not consider it further. A more detailed analysis will be presented in an upcoming work~\cite{oreshkov15}.
 
\section{Causal witnesses}\label{sec:causal_witness}

\subsection{Definition and characterization}

In this section we develop mathematical tools to identify, in the bipartite case, which process matrices are causally separable and which are not. In analogy with entanglement witnesses \cite{guhne09}, we call a hermitian operator $S$ a \emph{causal witness} (or \emph{witness}, simply) if\footnote{Note that the bound $0$ and the sign of the inequality are arbitrary; we choose them as in Eq.~\eqref{eq:causal_witness} for mathematical convenience.}
\begin{equation}\label{eq:causal_witness}
 \tr [S\, \wsep] \ge 0
\end{equation}
for every causally separable process matrix $\wsep$. This definition is motivated by the separating hyperplane theorem \cite{rockafellar70}: since the set of causally separable processes is closed and convex, for every causally nonseparable process matrix $W_{\text{ns}}$ there exists a causal witness $S_{W_{\text{ns}}}$ such that $\tr [S_{W_{\text{ns}}} W_{\text{ns}} ] < 0$.

To construct a witness for a given nonseparable process, we will start by characterizing the set of all causal witnesses in terms of linear constraints on a convex cone. This will allow us to cast the problem of finding a witness as an SDP problem. First, note that \eqref{eq:causal_witness} is equivalent to
   \begin{subequations}
   \begin{gather}
    \label{eq:wab} \tr [S\, W^{A \prec B}] \ge 0\quad \forall \, W^{A \prec B} \, , \\ 
    \label{eq:wba} \tr [S\, W^{B \prec A}] \ge 0\quad \forall \, W^{B \prec A} \, .
   \end{gather}
   \end{subequations}
  Let us focus on condition~\eqref{eq:wab}. Using Eq.~\eqref{eq:valid_W_AB} and noting that for any valid process matrix $W$, ${}_{B_O}W$ is a valid causally ordered process matrix compatible with the order $A \prec B$, one finds that~\eqref{eq:wab} is equivalent to
  \begin{equation}
   \tr \De{S ({}_{B_O}W)} \ge 0\quad \forall \, W \in \mathcal L_V, \ W \ge 0 \, .
  \end{equation}
  Thinking of the trace as the Hilbert-Schmidt inner product and noting that the map ${}_{B_O}\cdot$ is self-dual, we have that
  \begin{equation}
   \tr \De{S ({}_{B_O}W)} = \tr \De{({}_{B_O}S) \, W} ,
  \end{equation}
  and it is sufficient that ${}_{BO}S \ge 0$ for the right-hand-side to be non-negative for all valid $W$. An analogous argument shows that ${}_{A_O}S \ge 0$ is sufficient to satisfy condition \eqref{eq:wba}. We conclude that for $S$ to be a causal witness, it is sufficient that
  \begin{equation}\label{eq:witness_conditions}
    {}_{B_O}S \ge 0 \mathand {}_{A_O}S \ge 0. 
  \end{equation}
  Note also that adding an operator $S^\perp$ belonging to the orthogonal complement $\mathcal L_V^\perp$ of $\mathcal L_V$ to any witness $S$ gives another valid witness, since $\tr[(S+S^\perp) W] = \tr[S  W]$ for any valid process matrix $W$. It turns out that this suffices to completely characterize the set of causal witnesses, as stated in the following theorem:
  \begin{theorem}\label{th:witness}
  A hermitian operator $S\in {A_I}\otimes {A_O}\otimes{B_I}\otimes{B_O}$ is a \textit{causal witness} if and only if $S$ can be written as
  \begin{equation}
   S = S_P + S^\perp,
  \end{equation}
  where $S_P$ and $S^\perp$ are hermitian operators such that
  \begin{gather}
   {}_{B_O}S_P \ge 0, \quad
   {}_{A_O}S_P \ge 0,\quad
  L_V(S^\perp) = 0 \, .
  \end{gather}
\end{theorem}
The rather technical proof of this theorem is relegated to Appendix \ref{app:proof_thm1}. This theorem provides a characterization of the closed convex cone of causal witnesses $\mathcal S$.
  
Since $S^\perp$ does not change the expectation value $\tr\De{S W}$, it can freely be chosen to be for instance
\begin{equation}
S^\perp = L_V(S_P) - S_P,
\end{equation}
so that $S = L_V(S_P)$.
This has the effect of restricting witnesses to the subspace of valid processes $\mathcal L_V$, which have the following characterization:
  \begin{corollary}\label{th:witness_lv}
  A hermitian operator $S \in \mathcal L_V$  is a \textit{causal witness} if and only if there exists a hermitian operator $S_P \in {A_I}\otimes {A_O}\otimes{B_I}\otimes{B_O}$ such that $S = L_V(S_P)$, ${}_{B_O}S_P \ge 0$, and ${}_{A_O}S_P \ge 0$.
\end{corollary}
This restricted set of causal witnesses is also a closed convex cone, which we denote by $\mathcal S_V = \mathcal S \cap \mathcal L_V$.

One could define witnesses as belonging to $\mathcal S_{V}$ instead of $\mathcal S$, since both sets are as powerful in detecting causal nonseparability. However, some physically motivated witnesses, such as those presented in Section \ref{sec:game_witness} (for the tripartite case), do not belong to $\mathcal S_V$, which is why we use the more general definition that witnesses belong to $\mathcal S$.

\subsection{Finding causal witnesses}\label{sec:testing_separability}

The previous characterization of the convex cone of causal witnesses allows one to efficiently check the causal nonseparability of any process matrix $W$ through algorithms for semidefinite programming (SDP) \cite{nesterov87}. They output a causal witness if $W$ is causally nonseparable, and an explicit decomposition in terms of causally ordered process matrices otherwise.

The idea is simply to minimize $\tr[S\, W]$ over the cone of causal witnesses\footnote{In principle minimizing over $\mathcal S$ instead of $\mathcal S_V$ would lead to the same value for $\tr[S\, W]$, but this causes technical problems as explained in Appendix~\ref{sec:sdp_duality}.} $\mathcal S_V$, and check whether we obtain a negative value or not.
Note that in order to make $\tr[S \, W]$ lower bounded (to avoid getting a value $-\infty$ for causally nonseparable process matrices) a normalisation constraint on the witnesses has to be imposed. This normalisation is arbitrary -- any constraint that makes $\mathcal S_V$ compact suffices -- and different normalisation choices give rise to different interpretations for the value of $\tr[S\, W]$.  We shall normalise the witnesses by imposing that $\tr \De{S\,\Omega} \le 1$ for every (normalised) process matrix $\Omega$, for $-\tr[S\, W]$ can then be interpreted as a measure of causal nonseparability, as we shall see later in this subsection. In order to be able to use it in the SDP problem we still need to write this normalisation as a conic constraint. To do so, we extend the constraint $\tr \De{S\,\Omega} \le 1$ to non-normalised process matrices by linearity:
\begin{equation}
 \tr \De{S\,\Omega} \le \tr[\Omega]/d_O,
\end{equation}
which is equivalent to
\begin{equation}
 \tr \De{\de{\id/d_O-S} \, \Omega} \ge 0
\end{equation}
for all $\Omega \in \mathcal W$. Recalling that $S$ is assumed to be in $\mathcal S_V \subset \mathcal L_V$, this means that $\id/d_O - S \in \mathcal W^*_V := \mathcal W^* \cap \mathcal L_V$, where $\mathcal W^*$ is the dual cone of $\mathcal W$ -- that is, the cone of hermitian operators that have non-negative trace with process matrices.

To test the causal nonseparability of a given process matrix $W$, we are thus led to define the following SDP problem:
\begin{equation}\label{eq:sdp_witness}
\begin{gathered}
   \min \tr[SW] \\
  \text{s.t.}\quad S \in \mathcal S_V, \quad \id/d_O - S\ \in \mathcal W^*_V ,
\end{gathered}
\end{equation}
which is written explicitly in terms of positive semidefinite constraints in Appendix \ref{sec:explicit_sdp}.

If the solution of the SDP problem~\eqref{eq:sdp_witness} leads to a negative expectation value of $S$, one can conclude that $W$ is causally nonseparable, since SDP algorithms can be guaranteed\footnote{When the assumptions of the Duality Theorem (\ref{thm:duality}) are satisfied, which is the case for our SDP problems, as proven in Appendix~\ref{sec:sdp_duality}.} to find the optimal solution~\cite{nesterov87}. In such a case, the optimal solution $S^*$ provides an explicit witness to verify the causal nonseparability of $W$. On the other hand, if $\tr[S^* \, W] = 0$, one concludes that $W$ is causally separable, and an explicit decomposition of $W$ into causally ordered processes is given by the SDP problem dual to \eqref{eq:sdp_witness} (this can be seen explicitly from the representation of the SDP problem~\eqref{eq:sdp_dual} given in Appendix~\ref{sec:explicit_sdp}). As shown in Appendix~\ref{sec:sdp_duality}, this dual is
  \begin{equation}\label{eq:sdp_dual}
\begin{gathered}
 \min \tr[\Omega]/d_O \\
 \text{s.t.} \quad W + \Omega \in \wsepcone, \quad
 \Omega \in \mathcal W,
\end{gathered}
\end{equation}
where $\wsepcone$ is the cone of non-normalized causally separable process matrices, as previously defined. Furthermore, the optimal value $\tr[\Omega^*]/d_O$ of problem \eqref{eq:sdp_dual} is related to the optimal value $\tr[S^*W]$ of problem~\eqref{eq:sdp_witness} through
\begin{equation}\label{eq:duality_relation}
\tr [\Omega^*]/d_O = -\tr[S^* W].
\end{equation}
This gives an operational meaning to $-\tr[S^* W]$. As shown in Appendix \ref{sec:sdp_duality}, this quantity corresponds to the minimal $\lambda \ge 0$ such that
\begin{equation}
 \frac{1}{1+\lambda}\de{W + \lambda \, \widetilde\Omega}
\end{equation}
is causally separable, optimized over all valid, normalised processes $\widetilde\Omega$. In other words, it quantifies the resistance of $W$ to the worst-case noise. This is an analogue of the measure of entanglement called generalised robustness, which quantifies the resistance of the entanglement of a quantum state to worst-case noise \cite{steiner03}. It turns out that for our case the interpretation of $-\tr[S^* W]$ as a \emph{measure of causal nonseparability} is also tenable, as it respects some simple axioms that we propose in Appendix~\ref{sec:measures}. For this reason, we define the \textit{generalised robustness} of a process $W$ as
\begin{equation}\label{eq:def_generalised_robusntess}
 \gr(W) = -\tr(S^* W).
\end{equation}

Again in analogy with the case of entanglement measures, one can also define the \textit{random robustness} \cite{vidal99} of $W$ as is its resistance to ``white noise'', which can be defined as the process that sends maximally mixed states to each laboratory, independently of the local operations:
\begin{equation}\label{white}
 \id^\circ := \frac{\id}{d_{A_I}d_{B_I}}.
\end{equation}
The optimal witness with respect to random robustness can be found by solving an SDP problem analogous to~\eqref{eq:sdp_witness}:
\begin{equation}\label{eq:sdp_random_robustness}
\begin{gathered}
   \min \tr(SW) \\
  \text{s.t.}\quad S \in \mathcal S_V \, , \quad \tr(S\id^\circ) \le 1 \, ,
\end{gathered}
\end{equation}
whose dual is
  \begin{equation}\label{eq:sdp_random_robustness_primal}
\begin{gathered}
 \min \lambda \\
 \text{s.t.} \quad \lambda \ge 0 \, ,\quad W + \lambda \id^\circ \in \wsepcone \, ,
\end{gathered}
\end{equation}
and random robustness itself is defined as
\begin{equation}\label{eq:def_random_robusntess}
 \rr(W) = -\tr(S^* W) \, ,
\end{equation}
where $\tr(S^* W)$ is now the optimal value of the problem~\eqref{eq:sdp_random_robustness}.
This quantity can be used to compare witnesses in scenarios where white noise is an appropriate noise model, however, it cannot be interpreted as a proper measure of causal nonseparability, as it does not respect all the axioms we propose in appendix \ref{sec:measures} -- more specifically, it is not monotonous under local operations.

A geometrical interpretation of the results of this section is shown in Figure~\ref{fig:geometrical_witness}.

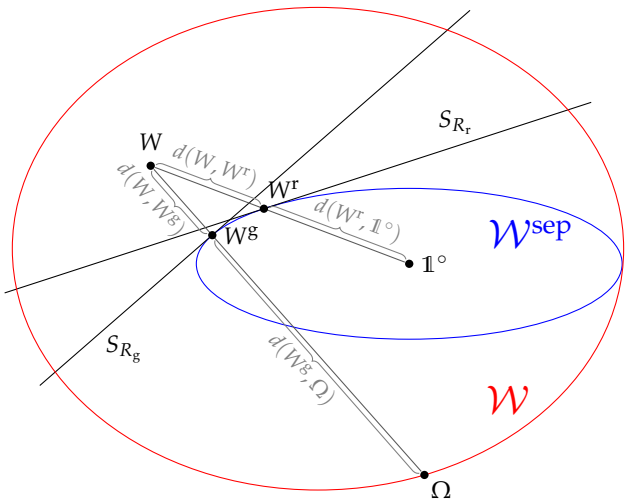
\begin{figure}
%
\begin{tikzpicture}[every node/.style={shape=circle,inner sep=1.1pt,fill=black},
                    every label/.style={fill=none},pin distance=2mm, 
    tangent/.style={
        decoration={
            markings,
            mark=
                at position #1
                with
                {
                    \coordinate (tangent point-\pgfkeysvalueof{/pgf/decoration/mark info/sequence number}) at (0pt,0pt);
                    \coordinate (tangent unit vector-\pgfkeysvalueof{/pgf/decoration/mark info/sequence number}) at (1,0pt);
                    \coordinate (tangent orthogonal unit vector-\pgfkeysvalueof{/pgf/decoration/mark info/sequence number}) at (0pt,1);
                }
        },
        postaction=decorate
    },
    use tangent/.style={
        shift=(tangent point-#1),
        x=(tangent unit vector-#1),
        y=(tangent orthogonal unit vector-#1)
    },
    use tangent/.default=1]
    \node [fill=none, blue] (wsepcone) at (1.6,0.4) {\Large $\wsepcone$};
        \node [fill=none, red] (wcone) at (1.3,-1.8) {\Large $\mathcal W$};
\node [label=0:{$\id^\circ$}] (id) at (0,0) {} ;
\node [label={$W$}] (w) at (-3.4,1.3) {} ; 
\draw [name path=ell,blue,tangent=0.402,tangent=0.465] (id) ellipse (2.8cm and 1cm);
\draw [black, use tangent] (-4.5,0) -- (3.6,0) node [pos=0.2, anchor=south east, fill=none,] {$S_{\rr}$};
\draw [black, use tangent=2] (-4.5,0) -- (3,0) node [pos=0.85,anchor=north west, fill=none] {$S_{\gr}$};
\draw [name path=ell2,red] (-1.2,0.2) ellipse (4.02cm and 3.2cm);
\node [label=-30:{$\Omega$}] (omega) at (0.2,-2.8) {} ; 
\draw [name path=idw, gray!70!black] (id) -- (w);
\draw [name path=omegaw, gray!70!black] (w) -- (omega);
\node[name intersections={of=idw and ell},label={45:{$W^\text{r}$}}] (is1) at (intersection-1) {}; 
\node[name intersections={of=omegaw and ell},label=-0:{$W^\text{g}$}] (is2) at (intersection-1) {}; 
\draw[decorate,decoration={brace,raise=1pt},gray]  (w) -- (is1) node [midway, above, sloped, fill=none,shape=rectangle, inner sep=4pt] {\footnotesize $d(W, W^\text{r})$};  
\draw[decorate,decoration={brace,raise=1pt},gray]  (is1) -- (id) node [pos=0.6, above, sloped, fill=none,shape=rectangle, inner sep=4pt] {\footnotesize $d(W^\text{r}, \id^\circ)$};  
\draw[decorate,decoration={brace,raise=1pt},gray]  (is2) -- (w) node [pos=0.75, below, sloped, fill=none,shape=rectangle, inner sep=4pt] {\footnotesize $d(W, W^\text{g})$};  
\draw[decorate,decoration={brace,raise=1pt},gray]  (omega) -- (is2)node [midway, below, sloped, fill=none,shape=rectangle, inner sep=4pt] {\footnotesize $d(W^\text{g}, \Omega)$};  
\end{tikzpicture}
%
\caption{
Here we schematically represent the set of normalised process matrices in $\mathcal W$ by the red ellipse and the set of normalised causally separable processes in $\wsepcone$ by the blue ellipse. Since the latter set is closed and convex, any causally nonseparable process $W$ is separated from it by a hyperplane, corresponding to an operator $S$ which we call a \textit{causal witness}. In the figure we represent two such causal witnesses, $S_{\gr}$ and $S_{\rr}$, that represent two different ways to quantify how far $W$ is from being causally separable. $-\tr(S_{\gr}W)$ measures the generalised robustness of $W$, which is its resistance to the worst-case noise $\Omega$. Geometrically, the generalised robustness of $W$ is given by the ratio of distances $d(W,W^\text{g})/d(W^\text{g},\Omega)$, where $W^\text{g}$ is the causally separable process closest to $W$ on the depicted line. In its turn, $-\tr(S_{\rr} W)$, the random robustness of $W$, is its resistance to the ``white noise'' $\id^\circ$. Geometrically, it is given by analogous ratio $d(W,W^\text{r})/d(W^\text{r},\id^\circ)$, where $W^\text{r}$ is again the causally separable process closest to $W$ on the depicted line. $S_{\gr}$ and $S_{\rr}$ are the optimal solutions of the SDP problems~\eqref{eq:sdp_witness} and~\eqref{eq:sdp_random_robustness}, respectively.
}
\label{fig:geometrical_witness}
\end{figure} 

\subsection{Implementing causal witnesses}\label{sec:measuring_witnesses}

Once a causal witness $S$ has been obtained for a given causally nonseparable process matrix $W$, a natural question is how to ``measure'' it, \ie, how to access the quantity $\tr[S \, W]$ -- and, in particular, check its sign -- experimentally.

To do so, note that as $S \in {A_I}\otimes {A_O}\otimes{B_I}\otimes{B_O}$ is a hermitian operator, it can always be decomposed as a linear combination of the form\footnote{In the decomposition~\eqref{eq:WitnDecomp}, $x,y,a$ and $b$ should \emph{a priori} simply be understood as labels for $M_{a|x}^{A_IA_O}$ and $M_{b|y}^{B_IB_O}$. We can however assume, without loss of generality, that ${}_{A_O} \big( \sum_a M_{a|x}^{A_IA_O} \big) \le \id^{A_IA_O} / d_{A_O}$ and ${}_{B_O} \big( \sum_b M_{b|y}^{B_IB_O} \big) \le \id^{B_IB_O} / d_{B_O}$ for all $x, y$ (we can indeed always include scaling factors in the coefficients $\gamma_{x,y,a,b}$). Introducing, when required, some complementary positive semidefinite operators $M_{\varnothing|x}^{A_IA_O}$ and $M_{\varnothing|y}^{B_IB_O}$ (with null coefficients $\gamma_{x,y,a,b}$), so that now ${}_{A_O} \big( \sum_a M_{a|x}^{A_IA_O} \big) = \id^{A_IA_O} / d_{A_O}$ and ${}_{B_O} \big( \sum_b M_{b|y}^{B_IB_O} \big) = \id^{B_IB_O} / d_{B_O}$, the sets $\{M^{A_I A_O}_{a|x}\}_a$ and $\{M^{A_I A_O}_{b|y}\}_b$ can then be interpreted as the CJ representation of instruments, for which $x,y$ are inputs and $a,b$ are outputs.}
\begin{equation}\label{eq:WitnDecomp}
S = \sum_{x,y,a,b} \gamma_{x,y,a,b} \ M_{a|x}^{A_IA_O} \otimes M_{b|y}^{B_IB_O} \, ,
\end{equation}
where $\gamma_{x,y,a,b}$ are real coefficients and $M_{a|x}^{A_IA_O}$ and $M_{b|y}^{B_IB_O}$ are positive semidefinite matrices that can be interpreted as the Choi-Jamio{\l}kowski representation of CP trace non-increasing maps (see Section~\ref{sec:formalism}). 

Expanding $\tr[S \, W]$,
\begin{eqnarray}
\tr[S \, W] &=& \! \sum_{x,y,a,b} \gamma_{x,y,a,b} \ \tr \big[ \big( M_{a|x}^{A_IA_O} \otimes M_{b|y}^{B_IB_O} \big) \, W \big], \qquad \label{eq:WitnDecomp_probas}
\end{eqnarray}
where according to the generalized Born rule~\eqref{born}, the terms $\tr \big[ \big( M_{a|x}^{A_IA_O} \otimes M_{b|y}^{B_IB_O} \big) \, W \big]$ represent the probabilities $P \big( M_{a|x}^{A_IA_O}, M_{b|y}^{B_IB_O} \big)$ that the maps $M_{a|x}^{A_IA_O}$ and $M_{b|y}^{B_IB_O}$ are realized.
We assume that these CP maps can be implemented even if the causal order of the parties is not well-defined. The quantity $\tr[S \, W]$ can thus in principle be implemented experimentally by estimating the probabilities $P \big( M_{a|x}^{A_IA_O}, M_{b|y}^{B_IB_O} \big)$ and combining them as in Eq.~\eqref{eq:WitnDecomp_probas}.

The decomposition~\eqref{eq:WitnDecomp} is not unique. Furthermore, as noted before we can add to any witness $S$ a term $S^\perp$ such that $L_V(S^\perp) = 0$ without changing its validity or its trace with any valid process. Hence, it actually suffices to find a decomposition for $S + S^\perp$ for some arbitrary $S^\perp$, implement the corresponding maps, and combine their statistics as above.

\subsection{Example}\label{sec:witness_example}

Let us now illustrate the above considerations on an explicit example. Ref.~\cite{oreshkov12} introduced the following process matrix, for a case where all incoming and outgoing systems of $A$ and $B$ are 2-dimensional (qubit) systems (\ie, $d_{A_I} = d_{A_O} = d_{B_I} = d_{B_O} = 2$):
\begin{equation} \label{W_OCB}
  W_\text{OCB} = \frac{1}{4} \left[ \id + \frac{\id^{A_I} Z^{A_O} Z^{B_I} \id^{B_O} + Z^{A_I} \id^{A_O} X^{B_I} Z^{B_O}}{\sqrt{2}} \right] \, , 
\end{equation}
where $Z$ and $X$ are the Pauli matrices, and tensor products are implicit.
One can easily check that $W_\text{OCB} \ge 0$, that $\tr [ W_\text{OCB} ] = 4 = d_O$, and that $W_\text{OCB}$ satisfies Eqs.~\eqref{eq:valid_bipartite_1}--\eqref{eq:valid_bipartite_3}, which ensures that it is indeed a valid process matrix.
It was shown that $W_\text{OCB}$ allows for a violation of a causal inequality (see Section~\ref{sec:causal_ineqs}), which implies that it is causally nonseparable.

\medskip

The concept of causal witnesses introduced here allows us to prove the causal nonseparability of $W_\text{OCB}$ more directly. Solving the SDP problem~\eqref{eq:sdp_random_robustness} with YALMIP~\cite{yalmip} and the solver MOSEK~\cite{mosek}, we obtained, up to numerical precision, the optimal witness with respect to random robustness
\begin{equation} \label{S_OCB}
  S_\text{OCB} = \frac{1}{4} \left[ \id - \big( \id^{A_I} Z^{A_O} Z^{B_I} \id^{B_O} + Z^{A_I} \id^{A_O} X^{B_I} Z^{B_O} \big) \right] \, .
\end{equation}
Applying it to $W_\text{OCB}$, we find that $-\tr[S_\text{OCB} \, W_\text{OCB}] = \rr(W_\text{OCB}) = \sqrt{2} -1  > 0$ (where $\rr(W_\text{OCB})$ is the random robustness as defined in Equation~\eqref{eq:def_random_robusntess}). This proves that $W_\text{OCB}$ is causally nonseparable.

This also implies that the process matrices of the form
\begin{equation}\label{W_OCBv}
  W_\text{OCB}(\lambda) = \frac{1}{1+\lambda} \, \de{ W_\text{OCB} + \lambda \, \id^\circ},
\end{equation}
are causally nonseparable for 
\begin{equation}
 \lambda < \rr(W_\text{OCB}) = \sqrt2 -1 
\end{equation}
(and their causal nonseparability is then witnessed by $S_\text{OCB}$).
For $\lambda \ge \sqrt2 -1$, $W_\text{OCB}(\lambda)$ is causally separable; the solution of the SDP problem~\eqref{eq:sdp_random_robustness_primal} provides an explicit decomposition for $W_\text{OCB}\de{\rr(W_\text{OCB})}$ (as can be seen when writing~\eqref{eq:sdp_random_robustness_primal} in a form similar to Eq.~\eqref{eq:sdp_dual_D2}), from which we can derive an explicit decomposition for all $W_\text{OCB}(\lambda)$ for $\lambda \ge \sqrt2 -1 $, as
\begin{equation}\label{W_OCBv_decomp}
  W_\text{OCB}(\lambda) = \frac{1}{2} W_\text{OCB}^{A \prec B}(\lambda) + \frac{1}{2} W_\text{OCB}^{B \prec A}(\lambda),
\end{equation}
where
\begin{align}
  W_\text{OCB}^{A \prec B}(\lambda) &:= \frac{1}{4} \left[ \id + \frac{\sqrt{2}}{1+\lambda} \ \id^{A_I} Z^{A_O} Z^{B_I} \id^{B_O} \right] \, ,  \\
  \quad W_\text{OCB}^{B \prec A}(\lambda) &:= \frac{1}{4} \left[ \id + \frac{\sqrt{2}}{1+\lambda} \ Z^{A_I} \id^{A_O} X^{B_I} Z^{B_O} \right]
\end{align}
are causally ordered process matrices. (Note that for $\lambda < \sqrt{2} -1$, $W_\text{OCB}^{A \prec B}(\lambda)$ and $W_\text{OCB}^{B \prec A}(\lambda)$ as defined above would not be positive semidefinite, which explains why Eq.~\eqref{W_OCBv_decomp} then fails to provide a valid causally separable decomposition of $W_\text{OCB}(\lambda)$.)

\medskip

To measure the witness $S_\text{OCB}$ and obtain the quantity $\tr[S_\text{OCB} \cdot W]$ experimentally, one can for instance decompose it in the following way: define, for $x,y,y',a,b = 0,1$, the CJ matrices
\begin{eqnarray}
M_{a|x}^{A_IA_O} &:=& \Big( \frac{\id {+} (-1)^a Z}{2} \Big)^{\!A_I} \!\otimes \Big( \frac{\id {+} (-1)^x Z}{2} \Big)^{\!A_O} , \label{MAA_OCB} \\
M_{b|y,y'=0}^{B_IB_O} &:=& \Big( \frac{\id {+} (-1)^b X}{2} \Big)^{\!B_I} \!\otimes \Big( \frac{\id {+} (-1)^{y+b} Z}{2} \Big)^{\!B_O} , \qquad \label{MBB0_OCB} \\
M_{b|y,y'=1}^{B_IB_O} &:=& \Big( \frac{\id {+} (-1)^b Z}{2} \Big)^{\!B_I} \!\otimes \frac{\id^{B_O}}{2} \, , \label{MBB1_OCB}
\end{eqnarray}
which represent measure-and-prepare maps (see Appendix~\ref{sec:cj}).
One can then check that
\begin{eqnarray}
S_\text{OCB} = 3 \cdot \id^\circ - 4 \, G_\text{OCB}
\end{eqnarray}
with
\begin{eqnarray}
G_\text{OCB} = \frac{1}{8} \sum_{x,y,a,b} \Big[ && \delta_{a,y} \ M_{a|x}^{A_IA_O} \! \otimes \! M_{b|y,y'=0}^{B_IB_O} \nonumber \\[-2mm]
&& + \ \delta_{b,x} \ M_{a|x}^{A_IA_O} \! \otimes \! M_{b|y,y'=1}^{B_IB_O} \Big] \, ,
\end{eqnarray}
where $\delta_{j,k}$ is the Kronecker delta.
Thus, one can compute $\tr [S_\text{OCB} \cdot W]$ by performing the maps above on $W$ and combining the probabilities $P\big(M_{a|x}^{A_IA_O}, M_{b|y,y'}^{B_IB_O}\big) = \tr [M_{a|x}^{A_IA_O} \! \otimes \! M_{b|y,y'}^{B_IB_O} \cdot W]$ as follows:
\begin{eqnarray}
&& \tr [ S_\text{OCB} \cdot W ] = 3 - 4 \, \tr [ G_\text{OCB} \cdot W ] \nonumber \\
&& \qquad = 3 - 4 \cdot \frac{1}{8} \sum_{x,y,a,b} \Big[ \delta_{a,y} \ P\big(M_{a|x}^{A_IA_O}\!, M_{b|y,y'=0}^{B_IB_O}\big) \nonumber \\[-3mm]
&& \hspace{3.2cm} + \delta_{b,x} \ P\big(M_{a|x}^{A_IA_O}\!, M_{b|y,y'=1}^{B_IB_O}\big) \Big] \, . \quad
\end{eqnarray}

\medskip

As one may recognize, the choice of CP maps in~\eqref{MAA_OCB}--\eqref{MBB1_OCB} is the same\footnote{Note that compared to Ref.~\cite{oreshkov12}, we exchanged in the present paper the notations $x,y$ and $a,b$ for inputs and outputs, so as to use here the same notations as most of the recent works on quantum and nonlocal correlations \cite{brunner14}. Furthermore, in~\cite{oreshkov12} the state sent out by $B$ when $y'=1$ was arbitrary, while here we fixed it to be $\id^{B_O}/2$.} as that considered in Ref.~\cite{oreshkov12}, so that the experimental procedure proposed here to measure the witness $S_\text{OCB}$ would be the same as that suggested in~\cite{oreshkov12} to violate a causal inequality.
The labels $x,y,y',a,b$ can be considered as inputs and outputs for the above maps (which indeed satisfy ${}_{A_O} \big(\sum_a M_{a|x}^{A_IA_O} \big) = \id^{A_IA_O} / d_{A_O}$ and ${}_{B_O} \big( \sum_b M_{b|y,y'}^{B_IB_O} \big) = \id^{B_IB_O} / d_{B_O}$ for all $x,y,y'$).
As it turns out, in the causal inequality of Ref.~\cite{oreshkov12} the probabilities $P(a,b|x,y,y') = P\big(M_{a|x}^{A_IA_O}, M_{b|y,y'}^{B_IB_O}\big)$ are actually combined in precisely the same way as above -- namely, $\tr [ G_\text{OCB} \cdot W ]$ above can be identified with the probability $p_\text{succ}$ of winning the corresponding ``causal game'',
\begin{equation}\label{eq:ocb_causal_inequality}
p_\text{succ} =  \frac12 \De{P(a=y|y'=0) + P(b=x|y'=1)},
\end{equation}
when the inputs $x, y, y' = 0, 1$ are given with equal probabilities. 

Remarkably, in this particular case the bounds of the causal witness $S_\text{OCB}$ and of the causal inequality \eqref{eq:ocb_causal_inequality} coincide, \ie, $\tr [ S_\text{OCB} \cdot W ] \ge 0$ if and only if $p_\text{succ} = \tr [ G_\text{OCB} \cdot W ] \le 3/4$, where $3/4$ is the upper bound on $p_\text{succ}$ for any causal correlation (as defined in Section~\ref{sec:causal_ineqs} below). Furthermore, the noise threshold below which the noisy process matrix $W_\text{OCB}(\lambda)$~\eqref{W_OCBv_decomp} can violate the causal inequality is the same as the threshold $\rr(W_\text{OCB})$ below which $W_\text{OCB}(\lambda)$ is causally nonseparable, as already noted in Ref.~\cite{brukner14}. This is however not a general property of causal witnesses and causal inequalities: similarly to the case of entanglement vs. quantum nonlocality and of entanglement witnesses vs. Bell inequalities~\cite{brunner14}, there exist causally nonseparable process matrices that cannot yield any violation of any causal inequality -- while there always exists a causal witness that detects their causal nonseparability. We will come back to this issue in Section~\ref{sec:causal_ineqs} below, with an explicit example in the tripartite case.

\section{Quantum control of causal order}\label{sec:Q_control}
\subsection{The quantum switch}
   It has recently been suggested that quantum computation can be extended beyond the framework of quantum circuits, which enforces a fixed order between the execution of quantum gates. The main idea is that the order in which gates are performed can be coherently controlled by a quantum system. The new resource that allows for such a control is the \textit{quantum switch}, first proposed in Ref.~\cite{chiribella09}. It works as follows: consider a two-qubit system, composed of a control and of a target qubit. Two parties $A$ and $B$ act on the target qubit with the unitaries $U_A$, $U_B$ respectively. If the control qubit is prepared in the state $\ket{0}$, $U_A$ is applied to the target before $U_B$, while if the control is in state $\ket{1}$ the two unitaries are applied in the reversed order. The global unitary, acting on both the target and control qubits, is thus
\begin{equation}
V(U_A,\,U_B) = \proj{0}\otimes U_BU_A + \proj{1}\otimes U_AU_B,
\label{unitaryswitch}
\end{equation}
where the first factor in each tensor product acts on the control system and the second factor acts on the target. For an initial state $\frac{\ket{0}+\ket{1}}{\sqrt{2}}\otimes\ket{\psi}$ of the control-target system, one gets, after applying $V$, the state $\frac{1}{\sqrt{2}}\big(\ket{0} \otimes U_BU_A\ket{\psi}+\ket{1} \otimes U_AU_B\ket{\psi}\big)$, which can be interpreted as having applied the two unitaries on the target in a ``superposition of orders''\footnote{Since any CP map can be purified to a unitary evolution by introducing an ancillary system and a projective measurement on some subsystem of the original system and ancilla, the notion of superposition of orders can be easily extended from unitary operations to arbitrary CP maps by introducing an ancillary register for each party.}.

Note that if the control system is discarded, one is left with the mixed state
\begin{equation}
\frac{1}{2}\left(U_BU_A \proj{\psi}U_A^{\dag}U_B^{\dag} + U_AU_B\proj{\psi}U_B^{\dag}U_A^{\dag}\right).
\label{mixed}
\end{equation}
This can be produced by randomly exchanging the order in which $U_A$ and $U_B$ are applied and thus can be seen as an equal mixture of causally ordered processes. To make the situation more interesting, we shall be led to introduce a third party, $C$, who can perform measurements on the control qubit (and possibly also on the target qubit) in order to define a causally nonseparable process (using the definition \eqref{eq:wsep_switch}) using quantum control of causal order.

\subsection{Process matrix representation of the quantum switch}
For our purposes, we can formally represent the quantum switch (with fixed input state) as a tripartite process matrix: the two parties $A$ and $B$ perform an arbitrary CP map each on the target qubit, while $C$ performs an arbitrary two-qubit POVM measurement on the resulting control-target state (with no outgoing system). The dimensions of input and output systems of the local laboratories are therefore
\begin{equation}
d_{A_I}=d_{A_O}=d_{B_I}=d_{B_O}=2,\quad d_{C_I}=4,\quad d_{C_O}=1.
\label{dimensions}
\end{equation}
For clarity, we shall divide $C$'s input space as $C_I=C_I^c\otimes C_I^t$, where $C_I^c$ and $C_I^t$ refer to the control and target qubits, respectively (with therefore $d_{C_I^c}=d_{C_I^t}=2$).

In order to describe the process matrix of the quantum switch, we are first going to make use of the ``pure'' version of the formalism, described in Appendix~\ref{sec:appendix_process_matrices}. An identity channel from a party's output space $A_O$ to another party's input space $B_I$ is described, as a process matrix, by the projector onto the ``process vector'' $\Ket{\id}^{A_OB_I}=\sum_{j=0,1} \ket{j}^{A_O}\ket{j}^{B_I}$. The situation where $A$ receives a state $\ket{\psi}$, performs an arbitrary operation on it, and sends the output directly to $B$ through an identity channel, who in turn sends the output of his operation to $C_I^t$, is represented by the process vector $\ket{\psi}^{A_I}\Ket{\id}^{A_OB_I}\Ket{\id}^{B_OC_I^t}$, see Appendix~\ref{sec:appendix_process_matrices}.
Then the quantum switch, with the control qubit initially in the state $\frac{\ket{0}+\ket{1}}{\sqrt{2}}$ and the target qubit in the state $\ket{\psi}$, is represented by the process matrix $\proj{w}$, where
\begin{eqnarray}
\label{eq:qswitch}
\ket{w} &=&  \frac{1}{\sqrt{2}}\left(\ket{\psi}^{A_I}\Ket{\id}^{A_OB_I}\Ket{\id}^{B_OC_I^t}\ket{0}^{C_I^c} \right. \nonumber \\[-2mm]
&& \qquad \ + \left. \ket{\psi}^{B_I}\Ket{\id}^{B_OA_I}\Ket{\id}^{A_OC_I^t}\ket{1}^{C_I^c} \right).
\end{eqnarray}
This can be checked by noting that
\begin{multline}
 \Bra{U_A^*}^{A_IA_O}\Bra{U_B^*}^{B_IB_O} \cdot \ket{w} = \\ 
\frac1{\sqrt{2}} \Big( \ket{0}^{C_I^c} \otimes \big( U_BU_A\ket{\psi} \big)^{C_I^t} + \ket{1}^{C_I^c} \otimes \big( U_AU_B\ket{\psi} \big)^{C_I^t} \Big),
\end{multline}
where $\Ket{U_A^*}^{A_IA_O} := \id \otimes U_A^* \Ket{\id}$ is the ``pure'' CJ representation of $U_A$, and similarly for $\Ket{U_B^*}^{B_IB_O}$ (and with $\Bra{U^*} = \Ket{U^*}^{\!\dagger\,}$); see Appendix~\ref{sec:cj}.
Note that the resulting state is the same as that obtained by applying~\eqref{unitaryswitch} to an initial state $\frac{\ket{0}+\ket{1}}{\sqrt{2}}\otimes\ket{\psi}$.

Note that the process~\eqref{eq:qswitch} itself is clearly causally nonseparable\footnote{Note that the results of Ref.~\cite{chiribella09} only show that the quantum switch cannot be realized by a circuit with a fixed order of gates, but the more general notion of causal (non)separability was not considered.}, since \begin{inparaenum}[\itshape i\upshape)]
\item it is a superposition of a pure process only compatible with the order $A \prec B \prec C$ and a pure process only compatible with the order $B \prec A \prec C$ and \item it is a projector onto a pure vector, thus it cannot be written as a nontrivial mixture of causally ordered processes. \end{inparaenum}

From Eq.~\eqref{eq:qswitch}, one finds (using the facts that $\tr_{C_I^t} (\KetBra{\id}^{B_OC_I^t}) = \id^{B_O}$ and $\tr_{C_I^t} (\KetBra{\id}^{A_OC_I^t}) = \id^{A_O}$) that by tracing out $C$, one gets
\begin{equation}
\tr_{C_I} \, \proj{w} = \tr_{C_I^cC_I^t} \, \proj{w} = \frac{1}{2} W^{A\prec B}+ \frac{1}{2} W^{B\prec A}, 
\end{equation}
where
\begin{align}
W^{A\prec B} &= \proj{\psi}^{A_I} \otimes \KetBra{\id}^{A_OB_I} \otimes \id^{B_O}, \\
W^{B\prec A} &= \proj{\psi}^{B_I} \otimes \KetBra{\id}^{B_OA_I} \otimes \id^{A_O},
\end{align}
are (bipartite) causally ordered process matrices; $\tr_{C_I} \, \proj{w}$ indeed describes the situation of Eq.~\eqref{mixed}.

For some information-processing tasks, the quantum switch is known to provide an advantage over causally ordered processes~\cite{chiribella12,araujo14}, even when $C$ ignores the target system and only measures the control system. We will thus restrict our attention to witnesses of the form $S^{C_I^c}\otimes \id^{C_I^t}$, which can simplify the analysis and the experimental implementation. The reduced process we will be dealing with is the partial trace of the quantum switch \eqref{eq:qswitch} over the target system:
  \begin{equation}\label{eq:switch_without_ci2}
   W_\text{switch} = \tr_{C_I^t} \proj{w}.
  \end{equation}
Note that the proof of causal nonseparability based on the purity of the switch does not extend to the reduced switch \eqref{eq:switch_without_ci2}, since it is not an extremal process. We will therefore use the framework of causal witnesses to show that the reduced switch is also causally nonseparable.

\section{Witnesses for the quantum switch}\label{sec:switch_witness}

Since the quantum switch is a tripartite process where ${d_{C_{O}}=1}$, we can use definition \eqref{eq:wsep_switch} to study its causal (non)separability.
In this tripartite situation, we will define causal witnesses to be the hermitian operators $S$ such that
\begin{equation}\label{eq:def_witness_switch}
 \tr\De{S\,W^{\text{sep}}} \ge 0
\end{equation}
for every causally separable processes $W^{\text{sep}}$ in the cone $\wsepcone_{3C}$.
The set of causal witnesses is thus the cone dual to $\wsepcone_{3C}$, which we denote by $\mathcal S_{3C}$, or $\mathcal S_{3C,V}$ when restricted to $\mathcal L_V$. The characterization of $\mathcal S_{3C}$ is given by the following theorem:
  \begin{theorem}\label{thm:switch_witness}
  A hermitian operator $S\in {A_I}\otimes {A_O}\otimes{B_I}\otimes{B_O}\otimes{C_I}\otimes{C_O}$ with ${d_{C_O} = 1}$ is a \textit{causal witness} if and only if $S$ can be written as
  \begin{equation}
   S = S^P_{ABC} + S^\perp_{ABC} = S^P_{BAC} + S^\perp_{BAC},
  \end{equation}
  where
  \begin{gather}
   S^P_{ABC} \ge 0, \quad L_{A \prec B \prec C}(S^\perp_{ABC}) = 0, \\
   S^P_{BAC} \ge 0, \quad L_{B \prec A \prec C}(S^\perp_{BAC}) = 0,
  \end{gather}
  with $L_{A \prec B \prec C}$ and $L_{B \prec A \prec C}$ as defined in Subsection~\ref{sec:causally_separable_3c}.
\end{theorem}
The proof is given in Appendix~\ref{app:proof_thm3}.
This characterization allows us to cast the problem of finding a witness for the quantum switch (or in fact for any process $W$ with $d_{C_O} = 1$) as an SDP problem analogous to \eqref{eq:sdp_witness}:
\begin{equation}\label{eq:sdp_witness_switch}
\begin{gathered}
   \min \tr(SW) \\
  \text{s.t.}\quad S \in \mathcal S_{3C,V}, \quad \id/d_O - S \in {\mathcal{W}^*_{3C,V}} \, ,
\end{gathered}
\end{equation}
where ${\mathcal{W}^*_{3C,V}} := {\mathcal{W}^*_{3C}} \cap \mathcal L_V$, with ${\mathcal{W}^*_{3C}}$ the dual of the cone ${\mathcal{W}_{3C}}$ of (non-normalized) tripartite process matrices with $d_{C_O} = 1$.

Analogously to problems \eqref{eq:sdp_witness}--\eqref{eq:sdp_dual}, the dual of \eqref{eq:sdp_witness_switch} writes
\begin{equation}\label{eq:sdp_dual_switch}
\begin{gathered}
 \min \tr[\Omega]/d_O \\
 \text{s.t.} \quad W + \Omega \in \wsepcone_{3C}, \quad
 \Omega \in \mathcal W_{3C},
\end{gathered}
\end{equation}
and the optimal values of \eqref{eq:sdp_witness_switch} and~\eqref{eq:sdp_dual_switch} respect the duality relation \eqref{eq:duality_relation}, which allows us to interpret $-\tr(S^*W)$ as generalised robustness also in this case. Furthermore, \eqref{eq:sdp_witness_switch} and \eqref{eq:sdp_dual_switch} respect the assumptions of the Duality Theorem, and therefore SDP algorithms can find their optimal solutions efficiently. We shall, however, omit the proofs, as they are simply a slight modification of the ones already presented in Appendix~\ref{sec:sdp_duality}.

\subsection{Optimal witness}\label{sec:SDP_witness_switch}

To find the optimal generalised robustness witness for the quantum switch we need to solve SDP problem~\eqref{eq:sdp_witness_switch} providing $W_\text{switch}$ from Eq.~\eqref{eq:switch_without_ci2} as an argument. Solving it using YALMIP and the solver MOSEK we obtain a witness $S_\text{optimal}$ numerically; the generalised robustness of the quantum switch is found to be
\begin{equation}
 \gr(W_\text{switch}) = -\tr S_\text{optimal} W_\text{switch} \approx 0.5454 \,.
\end{equation}
Later in this section we will compare this number to that obtained from non-optimal witnesses. For this purpose, we shall use the amount of worst-case noise tolerated by a witness, \ie, the amount of worst-case noise that can be added to the quantum switch before the witness can no longer detect its causal nonseparability. It should be clear that, when the said witness is optimal, this number reduces to the generalised robustness of the quantum switch.

\subsection{Chiribella's witness}\label{sec:game_witness}

In Ref.~\cite{chiribella12} Chiribella proposed an information-processing task for which the quantum switch had an advantage over causally ordered processes. We want to understand what this advantage means, and how it relates to causal nonseparability. For that we shall present a slightly modified version of his task and show how it can be understood as a causal witness.

Our version of the task is as follows: Alice (party $A$) receives a qubit in her lab, applies a unitary $U_A$ to it, and sends it away. Bob (party $B$) receives a qubit in his lab, applies a unitary $U_B$ to it, and sends it away. We assume that in each run of the experiment, $U_A$ and $U_B$ either commute or anticommute. Charlie (party $C$) receives a qubit in his lab, and makes a measurement on it to decide whether $U_A$ and $U_B$ commute or anticommute.

To construct a causal witness in relation to this task, we start with the Choi-Jamio{\l}kowski representation of the actions of the parties: Alice applying a unitary $U_A$, Bob applying a unitary $U_B$, and Charlie obtaining the result $\pm$ when measuring in the $\ket{\pm} = \frac{\ket{0}\pm\ket{1}}{\sqrt{2}}$ basis.
Using the CJ representations $\Ket{U_A^*}$ and $\Ket{U_B^*}$ of $U_A$ and $U_B$ (see Appendix~\ref{sec:cj}), the corresponding operator is 
\begin{equation}
 G^{U_A,U_B}_{\pm} = \KetBra{U_A^*}\otimes \KetBra{U_B^*} \otimes \ketbra{\pm}{\pm} \, .
\end{equation}
The witness corresponding to the task is obtained by averaging over the cases where Charlie obtains $+$ when Alice and Bob apply commuting unitaries, and the cases where Charlie obtains $-$ when Alice and Bob apply anticommuting unitaries:
\begin{equation}\label{eq:chiribella_game}
 G_\text{Chiribella} = \frac12 \int \mathrm{d}\mu_{\com} \, G^{U_A,U_B}_+ + \frac12\int \mathrm{d}\mu_{\ant} \, G^{U_A,U_B}_- \, ,
\end{equation}
where $\mathrm{d}\mu_{\com}$ is a measure over commuting unitaries, and $\mathrm{d}\mu_{\ant}$ is a measure over anticommuting unitaries (we assume here that the cases where $U_A$ and $U_B$ commute and anticommute each appear with probability $\frac12$). The probability of success in this task when the parties are using a strategy described by a process matrix $W$ is then
\begin{equation}
 p_\text{succ} = \tr [ G_\text{Chiribella} W]\, .
\end{equation}
It is easy to check that for any choice of measures $\mathrm{d}\mu_{\com},\mathrm{d}\mu_{\ant}$ the probability of success is $1$ when $W = W_\text{switch}$. The maximal probability of success for causally separable processes, however, depends crucially on the measures $\mathrm{d}\mu_{\com}$ and $\mathrm{d}\mu_{\ant}$. If we were to choose, for example, measures that only produce pairs of Pauli matrices, then there is a causally separable circuit\footnote{One such circuit, described in Ref.~\cite{chiribella12}, involves applying the Pauli unitaries to one half of a maximally entangled state and doing a measurement in the Bell basis.} that can decide the commutativity or anticommutativity with probability $1$.

To avoid this problem we will first choose measures that can produce \emph{any} pair of commuting or anticommuting unitaries (modulo global phases). Specifically, we choose the commuting measure $\mathrm{d}\mu_{\com}$ to pick up commuting unitaries of the form 
\begin{equation}
 U_A = U\begin{pmatrix} 1 & 0 \\ 0 & e^{i\theta_1} \end{pmatrix}U^\dagger \mathand U_B = U\begin{pmatrix} 1 & 0 \\ 0 & e^{i\theta_2} \end{pmatrix}U^\dagger,
\end{equation}
where $U$ is uniformly distributed according to the Haar measure, and $\theta_i$ are uniformly distributed in the interval $[0,2\pi]$. For the anticommuting measure $\mathrm{d}\mu_{\ant}$, we will use $U_A = VXV^\dagger$ and $U_B=VZV^\dagger$, where $V$ is also a Haar-random unitary (and $X$ and $Z$ are the Pauli matrices)\footnote{It turns out that with this choice of measures the witness $G_\text{Chiribella}$ is the same as we would obtain by translating the task from Ref.~\cite{chiribella12} directly into the language of causal witnesses; the only difference, then, is that in~\cite{chiribella12} the witness was decomposed in terms of measurements and repreparations, whereas we decomposed it using unitaries only.}.

With these measures $G_\text{Chiribella}$ turns out to be a valid causal witness, as the maximal probability of success for causally separable processes $p_\text{succ}^\text{sep}$ is bounded below one. To calculate it we need to solve the following SDP problem:
\begin{equation}
\begin{gathered}    \max \tr [G_\text{Chiribella}  W] \\
  \text{s.t.} \quad \tr W = d_O, \quad W \in \wsepcone_{3C}.
\end{gathered}
\end{equation}
Solving it with YALMIP and MOSEK, we obtain
\begin{equation}
 p_\text{succ}^\text{sep} \approx 0.9288 \, .
\end{equation}
The amount of worst-case noise that $G_\text{Chiribella}$ can tolerate is $0.0766$, which is much worse than the $0.5454$ tolerated by $S_\text{optimal}$.

An issue with $G_{\text{Chiribella}}$ is that it would take an infinite number of measurements to estimate each term of the sum in \eqref{eq:chiribella_game}. Furthermore $\mathrm{d}\mu_{\com}$ and $\mathrm{d}\mu_{\ant}$ were chosen arbitrarily, while it would be preferable to have a justification for the choice of a particular measure. Both problems are solved by restricting the unitaries $U_A$ and $U_B$ to come from a finite set. In this way we only need perform a finite number of measurements to estimate the witness, and it is possible to optimize the measures over commuting and anticommuting unitaries through SDP problems.

The best witness we found is obtained by choosing the following ten unitaries:
\begin{multline}
  \mathcal G = \{\id, X, Y, Z, \frac{X+Y}{\sqrt2},\frac{X-Y}{\sqrt2}, \\ \frac{X+Z}{\sqrt2},\frac{X-Z}{\sqrt2},\frac{Y+Z}{\sqrt2},\frac{Y-Z}{\sqrt2}\}
\end{multline}
($Y$ being the third Pauli matrix),
and defining the witness to be
\begin{equation}\label{eq:game_simple}
 G_\text{finite} = \sum_{i,j=1}^{10} q_{ij}^\com \, G^{U_i,U_j}_+ + q_{ij}^\ant \, G^{U_i,U_j}_- ,
\end{equation}
where $U_k \in \mathcal G$, and $q_{ij}^\com, q_{ij}^\ant$ are the input probability distributions over commuting and anticommuting unitaries, normalised such that $\sum_{i,j} ( q_{ij}^\com + q_{ij}^\ant ) = 1$.

To obtain the weights $q_{ij}^\com, q_{ij}^\ant$ and $p_\text{succ}^\text{sep}$ we solved an SDP problem presented in Appendix~\ref{sec:optimizing_chiribella}. We obtained
\begin{equation}
 p_\text{succ}^\text{sep} \approx 0.8690 \, ,
\end{equation}
and tolerance to worst-case noise $0.1507$, which is higher than $G_\text{Chiribella}$'s $0.0766$, but still lower than $S_\text{optimal}$'s $0.5454$.

We want to emphasize that the witnesses obtained in this subsection are equivalent to the ones defined through \eqref{eq:def_witness_switch} in the beginning of the present section -- the only difference being the arbitary choice of the causal bound being $\ge 0$ vs $\le p_\text{succ}^\text{sep}$. More precisely, let $G$ be a witness such that
\begin{equation}
 \tr(G\,W^\text{sep}) \le p_\text{succ}^\text{sep}
\end{equation}
for every (normalised) causally separable $W^\text{sep}$ and 
\begin{equation}
 T_0 \le \tr(G\,W) \le T_1
\end{equation}
for every (normalised) process matrix $W$.
Then
\begin{equation}
 S = \frac1{p_\text{succ}^\text{sep}-T_0}\de{ p_\text{succ}^\text{sep}\frac{\id}{d_O} - G}
\end{equation}
is a valid generalised robustness witness. Furthermore, if $S$ is the optimal witness for some process matrix $W$ that saturates the upper bound $\tr(G\,W) = T_1$, it follows that
\begin{equation}\label{eq:generalised_robustness_versus_p_succ}
 \gr(W) = - \tr[S W] = \frac{T_1-p_\text{succ}^\text{sep}}{p_\text{succ}^\text{sep}-T_0} \, .
\end{equation}
When $G$ is either $G_\text{finite}$ or $G_\text{Chiribella}$, we have that $T_0=0$ and $T_1=1$. And even though they are \emph{not} optimal witnesses for $W_\text{switch}$, the relationship between $p_\text{succ}^\text{sep}$ and resistance to worst-case noise is valid for them, \ie, for both $G_\text{finite}$ and $G_\text{Chiribella}$ the resistance to worst-case noise is equal to $1/p_\text{succ}^\text{sep} -1$, as given by~\eqref{eq:generalised_robustness_versus_p_succ}.

\section{Causal inequalities}
\label{sec:causal_ineqs}

The notion of causal separability considered above relies on the quantum description of the local laboratories. One may ask what are the constraints imposed by a definite causal structure regardless of the specific description, or even the physics governing the devices performing the local operations. To study such restrictions, we will make use of so-called \textit{causal inequalities}~\cite{oreshkov12}, which bound the possible correlations that can be established between events following a definite causal order. The violation of a causal inequality gives a stronger, device-independent signature of lack of causal order than the measurement of a witness. It is natural to ask whether it is possible to use the quantum switch to violate a causal inequality; we show below that this is not the case.

\subsection{Device-independent causal relations}

We still consider a multipartite scenario in which a set of $N$ parties $\{A^{i}\}_{i=1}^{N}$ are located in different, separated laboratories. Each party can perform operations and obtain measurement outcomes. Contrary to the previous case however, we do not consider here any particular physical description of what happens in each lab; the ``settings'' for the operations in the different laboratories and the measurement outcomes are labelled by some classical variables $x_{i}$ and $a_{i}$ (with $1 \le i \le N$), respectively; for simplicity we assume that the $x_{i}$'s and $a_{i}$'s take a finite number of values. Defining the vector of settings $\vec{x} = (x_{1}, \dots x_{N})$ and the vector of outcomes $\vec{a} = (a_{1}, \dots, a_{N})$, the device-independent description of the correlations established in such an experiment is encoded in the conditional probability $P(\vec{a}|\vec{x})$.

Causal inequalities~\cite{oreshkov12} are constraints on $P(\vec{a}|\vec{x})$ derived from the assumption that there exists an underlying causal structure defining the order between parties. To be more precise, let us represent the causal order in which the parties act by a permutation $\sigma$, defined such that party $i$ acts before party $j$ if and only if $\sigma(i) < \sigma(j)$. This leads to a total ordering of the parties, namely  $A^{\sigma(1)}\prec A^{\sigma(2)} \prec \ldots\prec A^{\sigma(N)}$. We then say that a probability distribution $P(\vec{a}|\vec{x})$ is \emph{compatible with the causal order $\sigma$} if no party signals to those before her\footnote{Note that this condition is strictly stronger than no-signalling to each individual party, since it is possible to signal to a group of parties without signalling to any individual party.}, namely if for every $i$ the marginal distribution
 \begin{equation}
 P(a_{\sigma(1)}, \ldots, a_{\sigma(i)}|\vec{x}) := \sum_{\substack{ a_{\sigma(j)} \\j > i}} P(\vec{a}|\vec{x})
 \end{equation}
 does not depend on the inputs $x_{\sigma(j)}$ with \mbox{$j > i$}; \ie,
 \begin{multline}
  P(a_{\sigma(1)}, \ldots, a_{\sigma(i)}|x_{\sigma(1)},\ldots,x_{\sigma(i)},x_{\sigma(i+1)},\ldots,x_{\sigma(N)}) \\ = P(a_{\sigma(1)}, \ldots, a_{\sigma(i)}|x_{\sigma(1)},\ldots,x_{\sigma(i)},x'_{\sigma(i+1)},\ldots,x'_{\sigma(N)}) \\ \forall \ x_{\sigma(j)},x'_{\sigma(j)} \, .
 \end{multline}
 A probability distribution that is compatible with at least one causal order $\sigma$ is said to be \emph{causally ordered}.
 
 More generally, we allow the parties to share randomness to agree on a specific order of sending signals between them before the inputs of the game are given to them. This allows for convex combinations of causally ordered probability distributions:
  \begin{equation}
    \label{eq:DI-causal-separability}
    P(\vec{a}|\vec{x}) = \sum_{\sigma} \, q_{\sigma} \, P_{\sigma}(\vec{a}|\vec{x}), \quad q_{\sigma} \geq 0, \quad \sum_{\sigma} q_{\sigma} = 1 \, ,
  \end{equation}
  where each $P_{\sigma}$ is compatible with a fixed order $\sigma$.
 These are still not the most general correlations compatible with the assumption of a definite causal structure, as one party could control the causal order of a set of parties in its future~\cite{baumeler13,baumeler14,oreshkov15}. Correlations compatible with this most general scenario of definite causal order are called simply \textit{causal}. In the bipartite case, the set of causal correlations forms a convex polytope, delimited by a finite number of facets that define causal inequalities~\cite{araujo14b}. The explicit definition of causal correlations in the general $N$-partite case is, however, rather cumbersome, and for the purposes of this article it will be enough to consider probability distributions of the form \eqref{eq:DI-causal-separability}, which is a sufficient (although not necessary) condition for causal separability.

As causally separable processes can only generate causal correlations, the violation of a causal inequality can also be used to detect the causal nonseparability of a process. While causal witnesses are \textit{device-dependent} and can only detect causal nonseparability if each party trusts her operation's implementation, causal inequalities are completely \textit{device-independent}: even if each party distrusts her laboratory, they can still detect causal nonseparability from the statistics of their experimental outcomes, if those violate a causal inequality. While for every causally nonseparable process there is causal witness that will detect its nonseparability, there are causally nonseparable processes cannot be used to violate any causal inequalities: in the next subsection we will prove that the quantum switch provides such an example. There is an analogy here with entanglement witnesses, which allow for a device-dependent way of detecting entanglement, and Bell inequalities, which provide a device-independent entanglement certification -- ``nonlocality''~ \cite{brunner14}. The important difference is that states violating Bell inequalities are physically implementable, while no example of a physically implementable process violating causal inequalities is known.

\subsection{Quantum control of orders and causal inequalities}

One might first wonder if the quantum switch allows for a causal inequality violation between $A$ and $B$ (such as the bipartite causal inequalities of Refs.~\cite{oreshkov12,araujo14b}); this is however clearly not the case since, as pointed out before, ignoring (\ie, tracing out) the third party $C$ makes the process matrix of the quantum switch causally separable.

One might still hope that the quantum switch can be used to violate a tripartite inequality (see e.g.~\cite{baumeler14}), explicitly involving party $C$; as it turns out, this is also impossible, as a consequence of the following theorem\footnote{A similar conclusion based on the same example has been obtained by Oreshkov and Giarmatzi independently of the other authors of this paper and is presented in Ref.~\cite{oreshkov15}}:

\begin{theorem}\label{th:septheorem}
Consider $N{+}1$ parties $\left\{A^1,\dots,A^N,C\right\}$ with settings $\left\{x_1,\dots,x_N,z\right\}$ and outcomes $\left\{a_1,\dots a_N,c\right\}$. If the marginal distribution
\begin{equation}\label{eq:marginal_P}
 P\left(\vec{a}|\vec{x},z\right):= \sum_c P\left(\vec{a},c|\vec{x},z\right)
\end{equation}
is such that
\begin{enumerate}[leftmargin=5mm]
	\item $P\left(\vec{a}|\vec{x},z\right) = P\left(\vec{a}|\vec{x}\right)$ -- \ie, it does not depend on $z$: $C$ does not signal to any other (group of) parties;
	\item $P\left(\vec{a}|\vec{x}\right) = \sum_\sigma \, q_\sigma \, P_\sigma\left(\vec{a}|\vec{x}\right)$, where $q_\sigma \geq 0$, $\sum_{\sigma} q_{\sigma} = 1$, and the probability distributions $P_\sigma$ are causally ordered,
\end{enumerate}
then the full $(N{+}1)$-partite probability distribution $P\left(\vec{a},c|\vec{x},z\right)$ is causal.
\end{theorem}
\begin{proof}
Using Bayes' rule and the assumptions of the theorem, we can write
\begin{align}
P\left(\vec{a},c|\vec{x},z\right) &= P\left(\vec{a}|\vec{x},z\right) \, P\left(c|\vec{a},\vec{x},z\right) \\
&= \sum_\sigma \, q_\sigma \, P_\sigma\left(\vec{a}|\vec{x}\right) \, P\left(c|\vec{a},\vec{x},z\right) \\
&= \sum_\sigma \, q_\sigma \, \widetilde{P}_\sigma\left(\vec{a},c|\vec{x},z\right),
\end{align}
where $\widetilde{P}_\sigma\left(\vec{a},c|\vec{x},z\right) := P_\sigma\left(\vec{a}|\vec{x}\right) \, P\left(c|\vec{a},\vec{x},z\right)$ is compatible with the order $A^{\sigma(1)}\prec\ldots\prec A^{\sigma(N)}\prec C$; this shows that $P\left(\vec{a},c|\vec{x},z\right)$ is causal.
\end{proof}

To see that the correlations generated by the quantum switch (Eq.~\eqref{eq:qswitch}) respect assumptions~\emph{1.} and~\emph{2.} of the previous theorem, let us calculate the marginal probability distribution defined in Eq.~\eqref{eq:marginal_P} through the generalized Born rule~\eqref{born}, when the three parties perform operations $M^{A_IA_O}_{a|x}, M^{B_IB_O}_{b|y}$ and $M^{C_I}_{c|z}$:
\begin{align}
 P(a,b|x,y,z) &= \sum_c \, \tr \De{ M^{A_IA_O}_{a|x} \otimes M^{B_IB_O}_{b|y} \otimes M^{C_I}_{c|z} \cdot \proj{w}} \nonumber \\
	   & \hspace{-15mm} = \tr \Big[ M^{A_IA_O}_{a|x} \otimes M^{B_IB_O}_{b|y} \otimes \Big( \sum_c M^{C_I}_{c|z} \Big) \cdot \proj{w} \Big] \, .
\end{align}
Since the third party $C$ has no output space ($d_{C_O} = 1$), then for any instrument $\{M^{C_I}_{c|z}\}$ we have $\sum_c M^{C_I}_{c|z} = \id^{C_I}$, so that
\begin{equation}
 P(a,b|x,y,z) = \tr \De{ M^{A_IA_O}_{a|x} \otimes M^{B_IB_O}_{b|y} \cdot W^{AB} }
\end{equation}
with
\begin{equation}\label{eq:reduced_switch}
 W^{AB} := \tr_{C_I} \, \proj{w} \, .
\end{equation}
This implies that $P(a,b|x,y,z)$ does not depend on $z$, as required.
As argued before, tracing out $C$ from the process matrix representing the quantum switch leads to a causally separable process matrix of the form
$W^{AB} = \frac{1}{2} W^{A\prec B} + \frac{1}{2} W^{B\prec A}$
with causally ordered process matrices $W^{A\prec B}$ and $W^{B\prec A}$, which can only generate causally ordered probability distributions $P_{A\prec B}$ and $P_{B\prec A}$. Hence, $P(a,b|x,y,z)$ can be decomposed as $\frac{1}{2} P_{A\prec B}(a,b|x,y,z) + \frac{1}{2} P_{B\prec A}(a,b|x,y,z)$, so that the second assumption of Theorem~\ref{th:septheorem} is also satisfied.

Therefore, the quantum switch represents an example of a causally nonseparable process that can only generate causal correlations, and hence cannot be used to violate any causal inequality\footnote{Note that Theorem~\ref{th:septheorem} implies that this is also true for the $N$-partite generalization of the quantum switch defined in \cite{araujo14}.}. It is noteworthy that all the examples of causally nonseparable processes for which a physical interpretation is known, including those generated by space-time superpositions~\cite{zych14}, fall into this category. This raises the question of whether causally nonseparable processes that do violate causal inequalities can be physically implemented at all.

\section{Conclusion}

The process matrix formalism was originally conceived as a rather speculative extension of quantum mechanics to possibly include the indefinite causal structures expected in a quantized theory of gravity \cite{hardy2007towards}. The results of this work show that, in fact, it is a natural framework to study a class of quantum resources which cannot be captured by the circuit model, but nonetheless are physically realizable and can provide powerful computational advantages. We have shown that the quantum switch, a recently demonstrated resource for quantum computation, can be conveniently represented as a causally non-separable process matrix. We have also presented causal witnesses that can verify the causal nonseparability of the switch. As they only require performing unitaries in a ``superposition of order'' and a final measurement of a control qubit, such witnesses can be easily implemented in quantum-optics setups, as the one employed in Ref.~\cite{procopio_experimental_2014}.  

The theory of causal witnesses developed here has close resemblances with the theory of entanglement witnesses. In both cases, one is interested in finding ways to certify that a resource is outside some convex set, the set of separable states in the latter case, that of causally nonseparable process matrices in the former case. Following this analogy, causal inequalities can be seen as the counterpart to the Bell inequalities, as they both provide device-independent tests regarding the existence of some classical variable: local hidden variables for measurement outcomes in one case, classical variables determining the causal order in the other. A significant difference between the two frameworks is that the problem of determining causal separability can be solved numerically with efficient algorithms, whereas characterizing entanglement has been proven to be an NP-hard problem \cite{gurvits2003classical}.

As one could expect from the analogy with entanglement, there exist causally nonseparable processes that cannot violate causal inequalities. What is striking, in the case of process matrices, is that a physical interpretation is known only for resources in this category. As one of the main open problems in this field is the characterization of physical process matrices, it is tempting to speculate whether the (im)possibility to violate causal inequalities could provide a useful guidance in this respect.
\subsection*{Acknowledgements}
We thank Michal Sedl{\'a}k for useful discussions. We acknowledge support from the European Commission project RAQUEL (No.~323970); the Austrian Science Fund (FWF) through the Special Research Program Foundations and Applications of Quantum  Science (FoQuS), the doctoral programme CoQuS, and Individual Project (No.~2462); FQXi; the John Templeton Foundation; the Templeton World Charity Foundation (grant TWCF 0064/AB38); the French National Research Agency through the `Retour Post-Doctorants' program (ANR-13-PDOC-0026); and the European Commission through a Marie Curie International Incoming Fellowship (PIIF-GA-2013-623456).

\appendix

\section{Details of the formalism}\label{sec:app_formalism}

Here we explore in more details the properties of the Choi-Jamio{\l}kowski (CJ) isomorphism and of the process matrix formalism.  Note that other existing definitions of the CJ isomorphism differ by a transposition or a partial transposition from the one given here, which follows the convention in~\cite{oreshkov12} and allows a direct identification of non-signaling processes with quantum states.

\subsection{Choi-Jamio{\l}kowski isomorphism}\label{sec:cj}

\paragraph{Pure CJ isomorphism.}
It is convenient to distinguish two versions of the CJ isomorphism: one for maps over density matrices and one for linear operators on pure state. The latter -- the ``pure CJ isomorphism'' -- can be represented via the ``double-ket'' notation~\cite{royer_wigner_1991, braunstein_universal_2000}. For a linear operator $A : \Hi^{A_I} \to \Hi^{A_O}$, we define\footnote{Superscripts on CJ vectors and CJ matrices indicate the systems they refer to (they may be omitted when the context makes it clear enough).}
\begin{equation}
\Ket{A^*}^{A_IA_O} := \id \otimes A^* \Ket{\id},
\label{cjpure}
\end{equation}
where $\Ket{\id} \equiv \Ket{\id}^{A_I A_I}:=\sum_j \ket{j}^{A_I}\otimes\ket{j}^{A_I} \in \Hi^{A_I} \otimes \Hi^{A_I}$ (with also, of course, the usual notation $\Bra{\id} = \Ket{\id}^\dagger$), and the complex conjugation $^*$ is defined with respect to the chosen orthonormal basis $\{\ket{j}^{A_I}\}$ of $\Hi^{A_I}$. The inverse map is given by
\begin{equation}
A\ket{\psi} = \big[\bra{\psi}^{A_I}\otimes \id^{A_O} \cdot \Ket{A^*}^{A_IA_O}\big]^*.
\label{inversecj}
\end{equation}
We say that $\Ket{A^*}$ is the CJ representation (or CJ vector) of $A$. The cumbersome complex conjugation in the definition allows us to have a simpler representation for the process matrix.
\medskip
\paragraph{Maximally entangled states and unitaries.}
Consider here the case where the input and output spaces have equal dimensions, $d_{A_I}=d_{A_O}$.
The state obtained by applying a local unitary to one subsystem of a maximally entangled state is also maximally entangled. in reverse, it is possible to generate any (bipartite) maximally entangled state by applying a local unitary to one subsystem of a reference maximally entangled state. Therefore, the CJ vector \textit{$\Ket{U^*}^{A_IA_O}=\id\otimes U ^*\Ket{\id}$ is maximally entangled if and only if $U$ is a unitary}. More explicitly, an operator
\begin{equation}
U \, = \, \sum_{jk} \, u_{jk} \, \ket{j}\!\bra{k}
\end{equation} 
is unitary if and only if $\sum_l u_{jl}u^*_{kl}=\sum_l u^*_{lk} u_{lj}=\delta_{j, k}$ for all $j,k$. One can check that this is also a necessary and sufficient condition for which 
\begin{equation}
\Ket{U^*}^{A_IA_O}= \sum_{jk}u^*_{jk} \ket{k}^{A_I}\ket{j}^{A_O}
\end{equation}
is maximally entangled.

\medskip
\paragraph{Measurement-preparation.}
Another useful linear operator is $\ketbra{\psi}{\phi}$, which describes the observation of an outcome $\ket{\phi}$ in a projective measurement, followed by the repreparation of a state $\ket{\psi}$. Plugging this into the definition \eqref{cjpure}, we find the CJ representation
\begin{equation}
\big| \,\left( \ketbra{\psi}{\phi} \, \right)^*\big\rangle\!\big\rangle^{A_IA_O}= \ket{\phi}^{A_I} \otimes \big(\ket{\psi}^* \big)^{A_O}.
\label{measure_prepare}
\end{equation}
Reciprocally, every pure product CJ vector represents a measurement-preparation operation.

An important particular case is when $\ket{\psi}=\ket{\phi}$, which corresponds to the ideal non-demolition von Neumann measurement:
\begin{equation}
\big| \,\left( \ketbra{\phi}{\phi} \,\right)^* \big\rangle\!\big\rangle^{A_IA_O} = \ket{\phi}^{A_I} \otimes \big(\ket{\phi}^* \big)^{A_O}.
\label{cjprojection}
\end{equation}

\medskip
\paragraph{Mixed CJ operators.}
For the general case of a linear map $\mathcal M^A: A_I \to A_O$, we define the CJ isomorphism as
\begin{equation}
 M^{A_IA_O} := \big[\mathcal I \otimes \mathcal M^A (\KetBra{\id})\big]^T \,.
\label{cjmixed}
\end{equation}
It is easy to verify that the definition \eqref{cjmixed} reduces to~\eqref{cjpure} for operators of the form $\mathcal M^A(\rho) = A\rho A^{\dag}$, \ie\ that, in such a case,
\begin{equation}
M = \KetBra{A^*} \, 
\label{consistency}
\end{equation}
(with $\Ket{A^*} \equiv \Ket{A^*}^{A_IA_O}$ and $\Bra{A^*} = \Ket{A^*}^\dagger$).

According to Choi's theorem \cite{choi_completely_1975}, a linear map $\mathcal M^A: A_I \to A_O$ is CP if and only if its CJ matrix is positive semidefinite, $M^{A_IA_O}\geq 0$. A characterization of the trace-preserving condition can be found using the inverse CJ isomorphism,
\begin{equation}
\mathcal M^A(\rho) = \big[\tr_{A_I} [\rho^{A_I} \otimes \id^{A_O} \cdot M^{A_IA_O}]\big]^T \,.
\label{inversemixed}
\end{equation}
By taking the trace of both sides of the equation, it can be readily verified that the map $\mathcal M^A$ is trace-preserving if and only if 
\begin{equation}
\tr_{A_O} M^{A_IA_O} = \id^{A_I}.
\label{CPTP}
\end{equation}
Note that a CP map can be part of an instrument only if it is \textit{trace-non-increasing}, a condition that translates to
\begin{equation}
\id^{A_I} -\tr_{A_O} M^{A_IA_O} \geq 0 \, .
\label{CPTNI}
\end{equation}

A useful example is the CPTP map $\mathcal M^A(\sigma)=\rho\tr\sigma$, which corresponds to the preparation of a (normalized) state $\rho$ independently of the input state $\sigma$. Its CJ representation is found to be
\begin{equation}
M^{A_IA_O}=\id^{A_I} \otimes (\rho^T)^{A_O} \, .
\label{preparation}
\end{equation}
A second relevant case is the CP (not trace-preserving) map that gives the probability of observing a POVM element $E$ in a measurement: $\mathcal M^A(\rho)=\tr [E \rho]$ (here $d_{A_O}=1$). Its CJ representation is simply
\begin{equation}
M^{A_I}=E^{A_I} \, .
\label{povm}
\end{equation}
Finally, the situation where a POVM element $E$ is measured on the state $\sigma$ in $A_I$ and a state $\rho$ is prepared in $A_O$ corresponds to the CP map $\mathcal M^A(\sigma)=\rho\tr\left[E\sigma\right]$, which has CJ representation
\begin{equation}
M^{A_IA_O}=E^{A_I} \otimes (\rho^T)^{A_O} \, .
\label{measurementpreparation}
\end{equation}

\subsection{Process matrices}\label{sec:appendix_process_matrices}

\setcounter{paragraph}{0}
Here we discuss in more detail some examples and properties of process matrices.

\medskip
\paragraph{Quantum states.}
Consider a bipartite process matrix of the form 
\begin{equation}
W^{A_IA_OB_IB_O}=\rho^{A_IB_I}\otimes \id^{A_OB_O}.
\label{bistate}
\end{equation}
According to the generalized Born rule, Eq.\ \eqref{born}, the probability for the two parties $A$ and $B$ to perform trace non-increasing CP maps with CJ matrices $M^{A_IA_O}$ and $M^{B_IB_O}$, respectively, is given by
\begin{align} \notag
P\big(M^{A_IA_O},M^{B_IB_O}\big) =&\tr\left[\big(M^{A_IA_O}\otimes M^{B_IB_O}\big) W \right],\\ \label{processstate}
 =&\tr\left[\big(E^{A_I}\otimes E^{B_I} \big) \rho^{A_IB_I} \right],
\end{align}
where $E^{A_I}:=\tr_{A_O}M^{A_IA_O}$ and $E^{B_I}:=\tr_{B_O}M^{B_IB_O}$. These operators are positive semidefinite and, because of Eq.~\eqref{CPTNI}, they can be completed to form a POVM. Thus Eq.~\eqref{processstate} corresponds to the probability of observing the POVM element $E^{A_I}\otimes E^{B_I}$ given the state $\rho^{A_IB_I}$; in other words, the process matrix \eqref{bistate} describes a bipartite state. Notice that a process matrix of this form does not allow signalling in either direction and therefore, being compatible with both $A\prec B$ and $B\prec A$, it is causally separable. This is irrespective of the state $\rho$, which can be entangled or separable. Note also the difference between the process matrix \eqref{bistate} and the CJ representation of state preparation, Eq.~\eqref{preparation}.
\medskip
\paragraph{Channels.}
Consider a bipartite situation where a party $A$ only performs state preparations, while the second party $B$ only performs measurements. In this case, the local laboratory of $A$ is characterized by a trivial input space, $d_{A_I}=1$, while $B$ has a trivial output space, $d_{B_O}=1$. The process matrix shared by $A$ and $B$, which represents here a quantum channel, is then defined on the space $A_O\otimes B_I \ni W$. The probability that $B$ observes a POVM element $E$ when $A$ prepares a state $\rho$ is given by
\begin{equation}
P\left(E|\rho\right)=\tr\left[ \left(\rho^T\right)^{A_O}\!\otimes E^{B_I} \, \cdot \, W^{A_OB_I}\right],
\end{equation}
where we used \eqref{preparation} and \eqref{povm} for the local operations. This is equivalent to saying that $B$ measures $E$ in the state $\tr_{A_O}\big[\big(\rho^T\big)^{A_O} \otimes \id^{B_I} \cdot W^{A_OB_I}\big]=\left[\tr_{A_O} \big[ \rho^{A_O} \otimes \id^{B_I} \cdot \big(W^T\big)^{A_OB_I}\big] \right]^T$. Comparing this with the inverse CJ transformation \eqref{inversemixed}, we find that the process matrix $W$ corresponds to a channel with CJ representation $W^T$. In other words, a channel $\mathcal C$ from $A_O$ to $B_I$ is represented by the process matrix
\begin{equation}
W^{A_OB_I}= \mathcal I \otimes \mathcal C (\KetBra{\id})
\label{channel}
\end{equation}
(with here $\Ket{\id} \equiv \Ket{\id}^{A_OA_O}$).
Note that the CJ representation of a channel, Eq.~\eqref{cjmixed}, differs by a transposition from the corresponding process matrix \eqref{channel}.
\medskip
\paragraph{Reduced process matrices.}
Given a multipartite process $W=W^{A^1_IA^1_O\dots A^N_IA^N_O}$ and a CPTP map for the $j$-th party with CJ matrix $M^{A^j_IA^j_O}$, we define the \textit{reduced process matrix} for the remaining $N-1$ parties, given $M^{A^j_IA^j_O}$, as
\begin{eqnarray}
&& \overline{W}(M^{A^j_IA^j_O}) \nonumber \\
&& :=\tr_{A^j_I A^j_O}\left[\big(\id^{A_I^1 A_O^1}\!\otimes\ldots M^{A^j_IA^j_O} \! \otimes\ldots \id^{A_I^N A_O^N}\big) \!\cdot\! W\right]. \qquad \quad
\label{reduced}
\end{eqnarray}
With the usual generalized Born rule \eqref{born}, the reduced process matrix gives the probability for the remaining $N-1$ parties to measure arbitrary CP maps, given that the $j$-th party performs $M^{A^j_IA^j_O}$. The explicit dependence of $\overline{W}$ on $M^{A^j_IA^j_O}$ accounts for the possibility of signalling: the remaining parties observe different probability distributions depending on the choice of CPTP map performed by party $j$. As an example, consider a process matrix of the form \eqref{channel}. If $A$ prepares a state $\rho$, the reduced process matrix for $B$ is
\begin{eqnarray}
\overline{W}^{B_I}(\rho)&=& \tr_{A_O}\left[\left(\rho^T\right)^{A_O}\!\otimes \id^{B_I} \, \cdot \,W^{A_O B_I}\right] \nonumber \\
&=&\sum_{jk} \, \bra{k}\rho^T\ket{j} \ \mathcal C\left(\ket{j}\!\bra{k}\right) = \mathcal C\left(\rho\right). \quad
\label{reda}
\end{eqnarray}
Thus, for a process that represents a channel from $A$ to $B$, the reduced process for $B$, given that $A$ prepares $\rho$, is simply the channel applied to $\rho$, as should be expected.

\medskip
\paragraph{Pure process matrices.}
In some cases, the process matrix turns out to be a rank-one projector: $W=\proj{w}$ for some ``process vector'' $\ket{w}$. If the CJ operators representing the local operations are also rank-one projectors, as is the case for unitaries and projective measurements followed by pure repreparations, it is convenient to work at the level of vectors and of probability amplitudes: given the local operations $A_1,\dots,A_N$ represented by the CJ vectors $\Ket{A_1^*}^{A_I^1 A_O^1},\dots,\Ket{A_N^*}^{A_I^N A_O^N}$, the overall probability amplitude is given (up to global phase, which we choose to be $0$) by
\begin{equation}
\Bra{A_1^*}^{A_I^1 A_O^1}\otimes\dots\otimes\Bra{A_N^*}^{A_I^N A_O^N} \cdot \ket{w}^{A_I^1 A_O^1 \dots A_I^N A_O^N}.
\label{amplitude}
\end{equation}
The probability is then obtained as the modulus square of the amplitude and conforms to the general expression \eqref{born}. Given that party $j$ performs the unitary $U_j$, the reduced process is clearly given by the partial scalar product
\begin{equation}
\id^{A_I^1 A_O^1}\otimes\dots\Bra{U_j^*}^{A_I^j A_O^j}\otimes\dots\id^{A_I^N A_O^N} \cdot \ket{w}^{A_I^1 A_O^1 \dots A_I^N A_O^N}.
\label{reducedpure}
\end{equation}

The process matrix describing a unitary channel $U$ from $A_O$ to $B_I$ is of particular interest. Using \eqref{channel}, we find that it is given by
\begin{equation}
\ket{w}^{A_OB_I}=\id \otimes U \Ket{\id} = \Ket{U}^{A_OB_I}.
\label{unitary}
\end{equation}
Note again the difference between this expression and the CJ representation \eqref{cjpure}. Generalizing this to a sequence of parties $A_1,\dots,A_N$, with the output of party $j$ connected to the input of party $j+1$ via the unitary $U_j$, we find
\begin{equation}
\ket{w}^{A^1_O\dots A^N_I}=\Ket{U_1}^{A^1_OA^2_I}\otimes\dots\otimes\Ket{U_N}^{A^{N-1}_OA^N_I}.
\label{manyunitaries}
\end{equation}

\section{Valid process matrices}\label{sec:valid_w}

The conditions for an operator $W \in  A_I \otimes A_O \otimes B_I \otimes B_O$ to be a valid process matrix were first found in Ref.~\cite{oreshkov12}, where they were formulated in a basis-dependent way. Here we derive the equivalent characterization of valid process matrices given in Eqs.~\eqref{W_pos}--\eqref{LV}; we formulate it in a basis-independent way, which we find to be more convenient for our purposes.

We present below the derivation in the bipartite case, and also write explicitly, for ease of reference, the characterization in the tripartite case. The $N$-partite case follows from a straightforward generalization.

\subsection{Bipartite process matrices}

Recall that a given operator $W \in  A_I \otimes A_O \otimes B_I \otimes B_O$ is a valid process matrix if and only if it yields, through the generalized Born rule~\eqref{born}, only well-defined probabilities -- that is, the probabilities must be non-negative and must sum up to $1$.

\medskip
\paragraph*{Non-negativity.}
As recalled previously, a map is completely positive if and only if its CJ representation is positive semidefinite.
Including the possibility that $A$ and $B$'s operations involve interactions with a (possibly entangled) ancillary system in a state $\rho^{A_I'B_I'}$, the non-negativity of probabilities is thus equivalent to\footnote{Note that ignoring the possibility of an ancillary system, one would only find that $W$ must be ``positive on pure tensors'' (with respect to the partition $A_I A_O / B_I B_O$) -- a class strictly larger than positive semidefinite matrices \cite{barnum05}.}
\begin{gather}
 \tr \left[ \big(M^{A_I' A_I A_O} \otimes M^{B_I' B_I B_O} \big) \cdot \big( \rho^{A_I'B_I'} \otimes W^{A_I A_O B_I B_O} \big) \right] \geq 0 \nonumber \\
 \forall \ \ M^{A_I' A_I A_O} \geq 0 \, , \ M^{B_I' B_I B_O}  \geq 0 \, , \ \rho^{A_I'B_I'} \geq 0 \, .
 \label{constr:non_neg}
\end{gather}

For the case where the ancillary spaces $A_I'$ and $B_I'$ are isomorphic to $A_I \otimes A_O$, and $M^{A_I' A_I A_O}$ and $\rho^{A_I'B_I'}$ are both projectors onto the maximally entangled state $\Ket{\id}^{A_I A_O / A_I A_O}:=\sum_{j,k} \ket{j,k}^{A_I A_O} \otimes \ket{j,k}^{A_I A_O}$ (where $\{\ket{j,k}^{A_I A_O}\}$ is an orthonormal basis of $\Hi^{A_I} \otimes \Hi^{A_O}$), we find that  the trace in~\eqref{constr:non_neg} is equal to $\tr [M^{B_I' B_I B_O} \cdot W^{A_I A_O B_I B_O}]$. Requiring that its value is non-negative for all $M^{B_I' B_I B_O}  \geq 0$ implies that $W$ must be positive semidefinite.

Reciprocally, $W \geq 0$ clearly implies that~\eqref{constr:non_neg} is satisfied. Hence, the non-negativity of probabilities is equivalent to $W$ being positive semidefinite, Eq.~\eqref{W_pos}.

\medskip
\paragraph*{Normalization.}
The fact that probabilities must sum up to $1$ for all instruments is equivalent to the constraint that the probability of realization of any CPTP map is $1$.
Now, recall that a CP map ${\cal M}^A : A_I\rightarrow A_O$ is trace-preserving if and only if its CJ matrix $M^{A_I A_O}$ satisfies $\tr_{A_O} M^{A_I A_O} = \id^{A_I}$ -- or equivalently, using the notation of Eq.~\eqref{def:notation_trace}, ${}_{A_O} M^{A_I A_O} = \id^{A_I A_O} / d_{A_O}$.
Ignoring here for simplicity the possible use of an ancillary system (which leads to the same conclusion\footnote{Taking into account a possible ancillary state $\rho^{A_I'B_I'}$, the same reasoning as below leads to Eqs.~\eqref{consrt:norm} and \eqref{eq:W_LA_LB_LAB}, with $W$ replaced by $\rho^{A_I'B_I'} \otimes W^{A_I A_O B_I B_O}$ (and where the definitions of the maps $L_A$ and $L_B$ should include the ancillary systems $B_I'$ and $A_I'$, resp.), which must hold for all $\rho^{A_I'B_I'}$ such that $\tr \rho^{A_I'B_I'} = 1$. One can easily check that these are indeed equivalent to~\eqref{consrt:norm} and \eqref{eq:W_LA_LB_LAB} in their original form.}), the normalization of probabilities is thus equivalent to
\begin{gather}
 \tr \left[ \big(M^{A_I A_O} \otimes M^{B_I B_O} \big) \cdot W^{A_I A_O B_I B_O} \right] = 1 \nonumber \\
\forall \ \ M^{A_I A_O} \geq 0 \, , \ M^{B_I B_O} \geq 0 \, , \  \nonumber \\  
\text{s.t.} \ \ {}_{A_O} M^{A_I A_O} = \id^{A_I A_O} / d_{A_O} \ , \ {}_{B_O} M^{B_I B_O} = \id^{B_I B_O} / d_{B_O} \, .
 \label{constr:normalization}
\end{gather}

First of all, note that the positivity of $M^{A_I A_O}$ and $M^{B_I B_O}$ is irrelevant here, since the set of positive semidefinite operators is a full dimensional subset of the space of hermitian operators\footnote{More precisely: for any hermitian operator $M^{A_I A_O}$, there always exists $\alpha > 0$ such that $M^{A_I A_O} + \alpha \, \id^{A_I A_O} / d_{A_O} \ge 0$. Assuming that ${}_{A_O} M^{A_I A_O} = \id^{A_I A_O} / d_{A_O}$, one can thus decompose $M^{A_I A_O}$ as $M^{A_I A_O} = (\alpha {+} 1) \, M_+^{A_I A_O} - \alpha \, M_-^{A_I A_O}$, with $M_+^{A_I A_O} = \frac{1}{\alpha + 1}(M^{A_I A_O} + \alpha \, \id^{A_I A_O} / d_{A_O}) \ge 0$ and $M_-^{A_I A_O} = \id^{A_I A_O} / d_{A_O} \ge 0$ satisfying ${}_{A_O} M_+^{A_I A_O} = {}_{A_O} M_-^{A_I A_O} = \id^{A_I A_O} / d_{A_O}$. Similarly, any hermitian operator $M^{B_I B_O}$ such that ${}_{B_O} M^{B_I B_O} = \id^{B_I B_O} / d_{B_O}$ can be decomposed as $M^{B_I B_O} = (\beta {+} 1) \, M_+^{B_I B_O} - \beta \, M_-^{B_I B_O}$, with $M_+^{B_I B_O} \ge 0$, $M_-^{B_I B_O} \ge 0$ and ${}_{B_O} M_+^{B_I B_O} = {}_{B_O} M_-^{B_I B_O} = \id^{B_I B_O} / d_{B_O}$. Note that the four pairs $(M_\pm^{A_I A_O},M_\pm^{B_I B_O})$ satisfy the assumptions of Eq.~\eqref{constr:normalization}, and therefore $\tr \left[ \big(M_\pm^{A_I A_O} \otimes M_\pm^{B_I B_O} \big) \cdot W^{A_I A_O B_I B_O} \right] = 1$. Expanding $\tr \left[ \big(M^{A_I A_O} \otimes M^{B_I B_O} \big) \cdot W^{A_I A_O B_I B_O} \right]$ using the decomposition just constructed, we find that its value is also $1$.
Hence, if Eq.~\eqref{constr:normalization} holds, then it also holds without the positivity constraints on $M^{A_I A_O}$ and $M^{B_I B_O}$; the converse is of course trivially true.}. The only relevant conditions are the normalization constraints.
Defining the maps ${}_{[1-A_O]}M = M - {}_{A_O} M$ and ${}_{[1-B_O]}M = M - {}_{B_O} M$, and noting that for any hermitian operators $x$ and $y$, the operators ${}_{[1-A_O]}x + \id / d_{A_O}$ and ${}_{[1-B_O]}y + \id / d_{B_O}$ satisfy the above normalization constraints (where from now on we are omitting the superscripts to reduce cluttering), we find that Eq.~\eqref{constr:normalization} is equivalent to
\begin{gather}
 \tr \Big[ \Big({}_{[1-A_O]}x + \id / d_{A_O}\Big) \otimes \Big({}_{[1-B_O]}y + \id / d_{B_O}\Big) W \Big] = 1 \nonumber \\
\forall \ \ x, \ y \, .
\label{constr:normalization_v2}
\end{gather}
For $x = y = 0$, this yields the normalization condition of Eq.~\eqref{normalization},
\begin{equation}
\tr [W] = d_{A_O} d_{B_O}  \, . \label{consrt:norm}
\end{equation}
For $y = 0$ and $x = 0$, respectively, this in turn implies 
\begin{gather}
\tr \De{ \de{{}_{[1-A_O]}x \otimes \id } W } = 0 \quad \forall \ x \, , \label{consrt:fA} \\
\tr \De{ \de{ \id \otimes {}_{[1-B_O]}y } W } = 0 \quad \forall \ y \, , \label{consrt:fB}
\end{gather}
which then imply
\begin{gather}
\tr \De{ \de{ {}_{[1-A_O]}x \otimes {}_{[1-B_O]}y } W } = 0 \nonumber \\
\forall \ x, \ y \, . \label{consrt:fAfB}
\end{gather}
Reciprocally, Eqs.~\eqref{consrt:norm}--\eqref{consrt:fAfB} clearly imply~\eqref{constr:normalization_v2}, so that these are equivalent to Eq.~\eqref{constr:normalization}.

Thinking of the trace as the Hilbert-Schmidt inner product
 \[ \prin{M}{W} = \tr [ M \cdot W ] \]
(for hermitian operators $M, W$) and noting that the maps ${}_{[1-A_O]} \cdot$ and ${}_{[1-B_O]} \cdot$ are self-dual, the conditions~\eqref{consrt:fA}--\eqref{consrt:fAfB} are equivalent to
\begin{gather}
 {}_{[1-A_O]} \big(\tr_{B_I B_O} W\big) = 0, \label{eq:b8} \\
 {}_{[1-B_O]} \big(\tr_{A_I A_O} W\big) = 0, \label{eq:b9} \\
 {}_{[1-A_O][1-B_O]} W = 0, \label{eq:b10}
\end{gather}
which we can rewrite as
\begin{gather}
 {}_{B_IB_O}W = {}_{A_OB_IB_O}W \, , \label{eq:b11} \\
 {}_{A_IA_O}W = {}_{A_IA_OB_O}W \, , \label{eq:b12} \\
 W = {}_{B_O}W + {}_{A_O}W - {}_{A_OB_O}W \, , \label{eq:b13}
\end{gather}
which are conditions \eqref{eq:valid_bipartite_1}--\eqref{eq:valid_bipartite_3}. 

Note that each condition \eqref{eq:b8}--\eqref{eq:b10} defines a linear subspace, and the intersection of these three linear subspaces is the smallest subspace that contains all valid bipartite process matrices, which we denote by\footnote{Note that although we do not write that explicitly, the projectors we define below (e.g. $L_V$), and of course the subspaces they define (e.g. $\mathcal L_V$), depend on the number of parties $N$.}
\begin{eqnarray}
\label{def:LV2_subspace}
\mathcal L_V = \big\{W \in A_I \! \otimes \! A_O \! \otimes \! B_I \! \otimes \! B_O \, | \, W = L_V(W)\big\}, \qquad
\end{eqnarray}
The projector onto this subspace, $L_V$, shall be used quite often in the paper, so it is useful to find an explicit expression for it. To do that, first we rewrite conditions \eqref{eq:b11}--\eqref{eq:b13} explicitly as projections onto subspaces, \ie, as
\begin{gather}
 W = L_A(W) \, , \ \ W = L_B(W) \, , \ \ W = L_{AB}(W) \, , \ \label{eq:W_LA_LB_LAB}
\end{gather}
where the projectors $L_A$, $L_B$, and $L_{AB}$ are given by
\begin{gather}
\label{eq:mapao} L_{A}(W) = W - {}_{B_IB_O}W + {}_{A_OB_IB_O}W  \, , \\
\label{eq:mapbo} L_{B}(W) = W - {}_{A_IA_O}W + {}_{A_IA_OB_O}W  \, , \\
 L_{AB}(W) = {}_{B_O}W + {}_{A_O}W - {}_{A_OB_O}W \, ,
\end{gather}
Since the three projectors above commute, the projector onto the intersection of their subspaces $L_V$ is given simply by the composition of $L_A$, $L_B$, and $L_{AB}$, \ie,
\begin{equation}
 L_V(W) = L_A \circ L_B \circ L_{AB}(W),
\end{equation}
which, after simplification, can be written as 
\begin{multline}\label{eq:w_subspace}
L_V(W) = {}_{A_O}W + {}_{B_O}W - {}_{A_OB_O}W \\- {}_{B_IB_O}W + {}_{A_OB_IB_O}W \\ - {}_{A_IA_O}W + {}_{A_IA_OB_O}W 
\end{multline}

Summing up, we conclude that an operator $W \in A_I \otimes A_O \otimes B_I \otimes B_O$ is a valid bipartite process matrix if and only if $W \ge 0$, $\tr W = d_{A_O} d_{B_O}$, and $W = L_V(W)$, as in Eqs~\eqref{W_pos}--\eqref{LV}.

\subsection{Tripartite process matrices}
\label{sec:appendix-tripartite-processes}
A similar reasoning leads to the conclusion that an operator $W \in A_I \otimes A_O \otimes B_I \otimes B_O \otimes C_I \otimes C_O$ is a valid tripartite process matrix if and only if $W \ge 0$, $\tr W = d_{A_O} d_{B_O} d_{C_O}$, and 
\begin{gather}
 W = L_A(W) \, , \quad
 W = L_B(W) \, , \quad
 W = L_C(W) \, , \nonumber \\
 W = L_{AB}(W) \, , \quad
 W = L_{AC}(W) \, , \quad
 W = L_{BC}(W) \, , \nonumber \\
 W = L_{ABC}(W) \, , \label{constr:LABC}
\end{gather}
where the maps $L_A$, $L_B$, $L_C$, $L_{AB}$, $L_{AC}$, $L_{BC}$, and $L_{ABC}$ are now commuting projectors onto linear subspaces of $A_I \otimes A_O \otimes B_I \otimes B_O \otimes C_I \otimes C_O$, defined by
\begin{gather*}
 L_{A}(W) = {}_{[1-(1-A_O)B_IB_OC_IC_O]}W \, , \\[1mm]
 L_{B}(W) = {}_{[1-(1-B_O)A_IA_OC_IC_O]}W \, , \\[1mm]
 L_{C}(W) = {}_{[1-(1-C_O)A_IA_OB_IB_O]}W \, , \\[1mm]
 L_{AB}(W) = {}_{[1-(1-A_O)(1-B_O)C_IC_O]}W \, , \\[1mm]
 L_{AC}(W) = {}_{[1-(1-A_O)(1-C_O)B_IB_O]}W \, , \\[1mm]
 L_{BC}(W) = {}_{[1-(1-B_O)(1-C_O)A_IA_O]}W \, , \\[1mm]
 L_{ABC}(W) = {}_{[1-(1-A_O)(1-B_O)(1-C_O)]}W \, ,
\end{gather*}
where we used the shorthand notation
\begin{equation}
 {}_{[\sum_X \alpha_X X]} W = \sum_X \alpha_X \cdot {}_X W
\end{equation}
for a sum over products of subsystems $X$ with coefficients $\alpha_X$ (and with ${}_1 W := W$).

The constraints in~\eqref{constr:LABC} are equivalent to \begin{gather}
 W = L_V(W) \, ,
\end{gather}
where the map $L_V$ is obtained here by composing the 7 maps $L_A$, $L_B$, $L_C$, $L_{AB}$, $L_{AC}$, $L_{BC}$, and $L_{ABC}$. One finds in this tripartite case, after simplification,
\begin{eqnarray}
 L_V(W) &=& {}_{\big[1 - (1 - A_O + A_I A_O) (1 - B_O + B_I B_O) (1 - C_O + C_I C_O)} \nonumber \\[-3mm]
 && \hspace{4cm} {}_{+ \ A_I A_O B_I B_O C_I C_O \big]} W \, , \nonumber \\
\end{eqnarray}
which defines a projector onto the linear subspace
\begin{gather}
\label{def:LV3_subspace}
\mathcal L_V = \big\{W \in A_I \!\otimes\! A_O \!\otimes\! B_I \!\otimes\! B_O \!\otimes\! C_I \!\otimes\! C_O \, | \, W = L_V(W)\big\} \, .
\end{gather}

\subsection{\texorpdfstring{$N$-partite process matrices}{N-partite process matrices}}

The generalization to the $N$-partite case is rather straightforward. We find that an operator $W\in A^1_I\otimes A^1_O\otimes \ldots \otimes A^N_I\otimes A^N_O $ is a valid $N$-partite process matrix if and only if $W \ge 0$, $\tr W = d_{A^1_O}\ldots d_{A^N_O}$, and for all $2^{N}-1$ non-empty subsets ${\cal X}$ of $\{1, \ldots, N\}$,
\begin{gather}
 W = L_{\cal X}(W) := {}_{\big[1 - \prod_{i \in {\cal X}} (1-A^i_O) \prod_{i \notin {\cal X}} A^i_I A^i_O \big]} W \, .
 \label{constr:LX}
\end{gather}

Note that the $2^{N}-1$ maps $L_{\cal X}$ are commuting projectors onto linear subspaces of $A^1_I\otimes A^1_O\otimes \ldots \otimes A^N_I\otimes A^N_O$.
The constraints~\eqref{constr:LX} are equivalent to
\begin{gather}
 W = L_V(W) \, ,
\end{gather}
where the map $L_V$ is obtained this time by composing the $2^{N}-1$ maps $L_{\cal X}$. More explicitly, one finds in the $N$-partite case the general expression
\begin{gather}
 L_V(W) = {}_{\big[1 \ - \ \prod_i (1 - A^i_O + A^i_I A^i_O) \ + \ \prod_i A^i_I A^i_O \big]} W \, , \label{eq:LVN}
\end{gather}
which again defines a projector onto the linear subspace
\begin{gather}
\label{def:LVN_subspace}
\mathcal L_V = \big\{W \in A^1_I\otimes A^1_O\otimes \ldots \otimes A^N_I\otimes A^N_O \ | \ W = L_V(W)\big\} \, .
\end{gather}

\begin{proof}
For any subset ${\cal X}$ of $\{1, \ldots, N\}$, define $P_{\cal X} = \prod_{i \in {\cal X}} (1-A^i_O) \prod_{i \notin {\cal X}} A^i_I A^i_O$. For ${\cal X} \neq {\cal X}'$, note that there exists (at least one) $i_0$ such that the product $P_{\cal X} P_{{\cal X}'}$ contains the factor $(1-A^{i_0}_O) A^{i_0}_I A^{i_0}_O$. Now, ${}_{[(1-A^{i_0}_O) A^{i_0}_O]}W = 0$, so that ${}_{[P_{\cal X} P_{{\cal X}'}]}W = 0$.

Developing $L_V$, we thus find $L_V(W) = {}_{[ \prod_{{\cal X} \neq \emptyset}(1 - P_{\cal X}) ]} W = {}_{[ 1 - \sum_{{\cal X} \neq \emptyset} P_{\cal X} ]} W = {}_{[ 1 - \sum_{\cal X} P_{\cal X} + \prod_i A^i_I A^i_O]} W$.
Now, one can write $\sum_{\cal X} P_{\cal X} = \sum_{k_1=0}^1 \ldots \sum_{k_N=0}^1 (1-A^1_O)^{k_1} (A^1_I A^1_O)^{1-k_1} \ldots (1-A^N_O)^{k_N} (A^N_I A^N_O)^{1-k_N} = \prod_i \sum_{k_i=0}^1 (1-A^i_O)^{k_i} (A^i_I A^i_O)^{1-k_i} = \prod_i (1-A^i_O + A^i_I A^i_O)$, from which Eq.~\eqref{eq:LVN} follows.
\end{proof}

\section{Characterization of bipartite causal witnesses}\label{app:proof_thm1}

To characterize below the set of causal witnesses in the bipartite case, we shall make use of some basic definitions and facts from convex analysis, which we state here without any proof; the interested reader can find them for instance in Sections 2, 14 and~16 of Ref.~\cite{rockafellar70}\footnote{Note a difference in language. This reference uses the polar cone $K^\circ$ instead of the dual $K^*$. They are simply related by $K^\circ = -K^*$.}.
 
\medskip  
  
Let $E$ be a vector space equipped with an inner product $\prin{\cdot}{\cdot}$, and let $E'$ be the space of all linear functionals on $E$. In the finite-dimensional case that interests us, $E'$ is isomorphic to $E$.
  
  \begin{enumerate}
   \item A subset $\mathcal K$ of $E$ is a convex cone if and only if for every $x,y\in \mathcal K$ we also have that $\lambda x + \mu y \in \mathcal K$, for any $\lambda,\mu >0$. 
   \item Let $\mathcal K \subseteq E$ be a convex cone. Then its dual cone $\mathcal K^* \subseteq E'$ is defined as
   \begin{equation}
  \mathcal K^*  = \{ x^* \in E' \midset \prin{x^*}{x} \ge 0\quad\forall x \in \mathcal K\} \, .
  \end{equation}
  \item  The dual of a linear subspace $\mathcal L \subseteq E$ is its orthogonal complement:
  \begin{equation}\label{eq:orthogonal_complement}
  \mathcal L^*= \mathcal L^\perp = \{x^* \in E' \midset \prin{x^*}{x} = 0 \quad \forall x \in \mathcal L\}.
  \end{equation}
  \item Let $\mathcal K_1,\mathcal K_2 \subseteq E$ be closed convex cones that contain the origin. Then
  \begin{subequations}\label{eq:dual_properties}
  \begin{gather}
   [\conv(\mathcal K_1\cup \mathcal K_2)]^* = \mathcal K_1^* \cap \mathcal K_2^* \, , \\
   (\mathcal K_1\cap \mathcal K_2)^* = \conv(\mathcal K_1^* \cup \mathcal K_2^*) \, ,
  \end{gather}
  where $\conv$ denotes the convex hull.
\end{subequations}
 \end{enumerate}
 
We shall furthermore use below the following characterization of bipartite causally separable process matrices:  
\begin{lemma}\label{th:causally_separable}
A given matrix $W \in A_I \otimes A_O \otimes B_I \otimes B_O$ is a valid causally separable process matrix if and only if $\tr W = d_O$, $W \in \mathcal L_V$ (\ie, it satisfies Eqs.~\eqref{eq:valid_bipartite_1}--\eqref{eq:valid_bipartite_3}), and it can be written as
 \begin{gather}
W = W^{A \prec B} + W^{B \prec A} \label{decomp:causally_separable_1} \\[1mm]
 \textrm{with} \quad
 W^{A \prec B} \ge 0 \, , \quad
 W^{A \prec B} = {}_{B_O}W^{A \prec B} \, , \label{decomp:causally_separable_2} \\ 
 \quad \qquad W^{B \prec A} \ge 0 \, , \quad   
 W^{B \prec A} = {}_{A_O}W^{B \prec A} \, .  \label{decomp:causally_separable_3}
 \end{gather}
\end{lemma}

\begin{proof}
The ``only if'' direction is straightforward (simply replace $q \, W^{A \prec B} \to W^{A \prec B}$ and $(1{-}q) \, W^{B \prec A} \to W^{B \prec A}$ to go from~\eqref{def:caus_sep} to~\eqref{decomp:causally_separable_1}, so that $W^{A \prec B}$ and $W^{B \prec A}$ in~\eqref{decomp:causally_separable_1} are not normalized).

To see that the converse also holds, first note that $W^{A \prec B} \ge 0$ and $W^{B \prec A} \ge 0$ imply that $W \ge 0$, so that $W$ is indeed a valid process matrix. Note furthermore that $W^{B \prec A} = {}_{A_O}W^{B \prec A}$ implies that ${}_{B_IB_O}W^{B \prec A} = {}_{A_OB_IB_O}W^{B \prec A}$, \ie, that $W^{B \prec A}$ satisfies~\eqref{eq:valid_bipartite_1}. Since $W \in \mathcal L_V$ also satisfies~\eqref{eq:valid_bipartite_1}, so does $W^{A \prec B} = W - W^{B \prec A}$. Similarly, $W^{A \prec B} = {}_{B_O}W^{A \prec B}$, together with the assumption that $W \in \mathcal L_V$, implies that both $W^{A \prec B}$ and $W^{B \prec A}$ satisfy~\eqref{eq:valid_bipartite_2}. Lastly, $W^{A \prec B} = {}_{B_O}W^{A \prec B}$ and $W^{B \prec A} = {}_{A_O}W^{B \prec A}$ directly imply that both $W^{A \prec B}$ and $W^{B \prec A}$ satisfy~\eqref{eq:valid_bipartite_3}. All in all, this shows that $W^{A \prec B}$ and $W^{B \prec A}$ are, up to normalization (which can easily be dealt with as above so as to recover the form~\eqref{def:caus_sep}), valid causally ordered process matrices.
\end{proof}

\medskip  
  
  We are now in a position to prove Theorem~\ref{th:witness}:
\begin{proof}[Proof of Theorem~\ref{th:witness}]
  We want to characterize the set of all hermitian operators $S \in {A_I}\otimes {A_O}\otimes{B_I}\otimes{B_O}$ such that
  \begin{equation}
\tr[S \wsep] \ge 0   
  \end{equation}
   for all causally separable process matrices $\wsep$. Note that the condition $\tr[\wsep] = d_O$ is not relevant for the characterization of the witnesses, so we shall lift it. Without this restriction the set of (non-normalized) causally separable process matrices becomes a convex cone, which we denote by $\wsepcone$. If we consider the duality relations with respect to the Hilbert-Schmidt inner product in the space ${A_I}\otimes {A_O}\otimes{B_I}\otimes{B_O}$, then the convex cone of causal witnesses $\mathcal S$ is the dual cone of $\wsepcone$. To characterize it we are going to use the representation of $\wsepcone$ that follows from Lemma~\ref{th:causally_separable}:
   \begin{equation}
  \wsepcone = \conv\!\big[(\mathcal P \cap \mathcal L_{B_O}) \cup (\mathcal P \cap \mathcal L_{A_O} )\big] \cap \mathcal L_V,
   \end{equation}
    where $\mathcal P$ is the self-dual cone of positive semidefinite matrices and $\mathcal L_{B_O}$ and $\mathcal L_{A_O}$ are the linear subspaces 
   \begin{gather}
    \mathcal L_{B_O} = \{W \midset W = {}_{B_O}W\}, \\
    \mathcal L_{A_O} = \{W \midset W = {}_{A_O}W\},
   \end{gather}
 and with $\mathcal L_V$ defined in Eq.~\eqref{def:LV2_subspace}. Their orthogonal complements within the subspace $A_I \otimes A_O \otimes B_I \otimes B_O$ are simply given by
   \begin{gather}
    \mathcal L_{B_O}^\perp = \{S \midset {}_{B_O}S = 0 \}, \\
    \mathcal L_{A_O}^\perp = \{S \midset {}_{A_O}S = 0 \}, \\
    \mathcal L_V^{\perp} = \{S \midset L_V(S) = 0 \}.
   \end{gather}
  
    Taking then the dual of ${\mathcal{W}^{\text{sep}}}$ using the duality relations \eqref{eq:orthogonal_complement}--\eqref{eq:dual_properties}, we get that the cone of causal witnesses is
    \begin{equation}
    S =  \conv \Big[ \big[(\mathcal P \cap \mathcal L_{B_O})^* \cap (\mathcal P \cap \mathcal L_{A_O})^*\big] \cup {\mathcal L}_V^{\perp}\Big]. \label{eq:calS}
    \end{equation}
  Focusing on $(\mathcal P \cap \mathcal L_{B_O})^*$, using again the duality relations \eqref{eq:orthogonal_complement}--\eqref{eq:dual_properties}, we see that
  \begin{align}
   (\mathcal P \cap \mathcal L_{B_O})^* &= \conv (\mathcal P \cup \mathcal L_{B_O}^\perp)  
  \nonumber \\ &= \{S_+ + S_0 \ | \  S_+ \ge 0, \, {}_{B_O}S_0 = 0 \} \nonumber \\
    &= \{ S \midset {}_{B_O} S \ge 0 \},
  \end{align}
  where the last equality is stating the fact that $S = S_+ + S_0$ with $S_+ \ge 0$ and ${}_{B_O}S_0 = 0$ if and only if ${}_{B_O}S \ge 0$. To see that this is true, let $S$ be such that ${}_{B_O}S \ge 0$. Define then $S_+ = {}_{B_O}S$ and $S_0 = S - {}_{B_O}S$. Then $S = S_+ + S_0$, $S_+ \ge 0$ by assumption, and ${}_{B_O}S_0 = 0$ since the map ${}_{B_O}\cdot$ is a projector. To prove the other direction, let $S = S_+ + S_0$ with $S_+ \ge 0$ and ${}_{B_O}S_0 = 0$. The map ${}_{B_O}\cdot$ being positive, it follows that ${}_{B_O}S = {}_{B_O}S_+ \ge 0$.
  
  Similarly,
    \begin{equation}
    (\mathcal P \cap \mathcal L_{A_O})^* = \{ S \midset {}_{A_O}S \ge 0 \}. \label{eq:AOSpos}
    \end{equation}
  Putting Eqs.~\eqref{eq:calS}--\eqref{eq:AOSpos} together, we see that a causal witness can be written as $S = S_P + S^\perp$ with ${}_{B_O}S_P \ge 0, {}_{A_O}S_P \ge 0$, and $L_V(S^\perp) = 0$.
\end{proof}

\section{Explicit positive semidefinite constraints}\label{sec:explicit_sdp}

For the convenience of the reader, we present the SDP problems \eqref{eq:sdp_witness} and \eqref{eq:sdp_dual} with all conic constraints rewritten in terms of the positive semidefinite cone, to facilitate implementation. To rewrite \eqref{eq:sdp_witness}, we need a characterisation of the cones $\mathcal S_V$ and $\mathcal W^*_V$. The first one is given by Corollary~\ref{th:witness_lv}. The second one is obtained as follows: since $\mathcal W = \mathcal P \cap \mathcal L_V$, we have that
\begin{align}
 \mathcal W^*_V &= \mathcal W^* \cap \mathcal L_V \\
	        &= (\mathcal P \cap \mathcal L_V)^* \cap \mathcal L_V \\
	        &= \conv(\mathcal P \cup \mathcal L_V^\perp) \cap \mathcal L_V \\
	        &=  \{L_V(\Sigma_P)|\Sigma_P \ge 0\} \, , \label{eq:w_dual}
\end{align}
where Equation \eqref{eq:w_dual} follows from an argument analogous to the one used to derive Corollary~\ref{th:witness_lv}. 

With this characterisation, the SDP problem \eqref{eq:sdp_witness} then becomes  
\begin{equation}
\begin{gathered}
   \min \tr(SW) \\
  \text{s.t.}\quad S = L_V(S_P) \, ,\  {}_{A_O}S_P \ge 0 \, ,\  {}_{B_O}S_P \ge 0 \, , \\ 
  \id/d_O - S = L_V(\Sigma_P) \, , \ \Sigma_P \ge 0 \, .
\end{gathered}
\end{equation}
To rewrite the SDP problem \eqref{eq:sdp_dual}, we use the characterisation of $\wsepcone$ given in Lemma~\ref{th:causally_separable}:
\begin{equation}
\begin{gathered}
 \min \tr(\Omega)/d_O \\
 \text{s.t.} \quad W + \Omega  = W^{A \prec B} + W^{B \prec A} \, , \\
 W^{A \prec B} \geq 0 \, , \quad  W^{A \prec B} = {}_{B_O}W^{A \prec B} \, , \\
 W^{B \prec A} \geq 0 \, , \quad  W^{B \prec A} = {}_{A_O}W^{B \prec A} \, , \\
 \Omega \ge 0 \, ,\quad \Omega = L_V(\Omega) \, .
\end{gathered}
\label{eq:sdp_dual_D2}
\end{equation}
Note that we could use directly the definition of $\wsepcone$ from Section \ref{sec:causally_separable_2}, which would give us a slightly more complicated SDP problem.

\section{Duality for conic problems}\label{sec:sdp_duality}

In this appendix we show that the two problems defined in Section~\ref{sec:testing_separability} are SDP problems and they are dual to each other. We show, furthermore, that the Duality Theorem applies to them, which implies that the optimal solutions can be found efficiently and that Equation~\eqref{eq:duality_relation} holds. 

Let us first recall the definitions of primal and dual conic problems (Definition~4.2.1 in~\cite{nesterov87}), of which SDP problems are a particular case:
\begin{definition}\label{def:sdp_sovietic}
 Let $E$ be a finite-dimensional vector space, $\mathcal K$ a closed convex pointed cone in $E$ with a nonempty interior, and $\mathcal L$ a linear subspace of $E$. Let also $b \in E$ and $c \in E'$. The data $E$, $\mathcal K$, $\mathcal L$, $b$, and $c$ define a pair of conic problems
\begin{align*}
(P):& \quad \min \ \prin{c}{x} \quad \textup{s.t.} \quad x \in \mathcal K \cap (\mathcal L+b), \\
(D):& \quad \min \ \prin{y}{b} \quad \textup{s.t.} \quad y \in \mathcal K^* \cap (\mathcal L^\perp+c), 
\end{align*}
where $\mathcal K^* \subset E'$ is the cone dual to $\mathcal K$, $\mathcal L^\perp \subset E$ is the orthogonal complement to $\mathcal L$, $\mathcal L+b \subset E$ and $\mathcal L^\perp+c \subset E'$ are affine subspaces.
(P) and (D) are called, respectively, the primal and dual problems associated with the above data.
\end{definition}

We want our SDP problems to measure how much worst-case noise needs to be added to a given process matrix $W$ to make it causally separable, \ie, the minimal $\lambda \ge 0$ for which 
\begin{equation}
 \frac{1}{1+\lambda}\de{W + \lambda \, \widetilde\Omega}
\end{equation}
is a causally separable process, optimized over all valid (normalised) processes $\widetilde\Omega$. First note that we can get rid of the quadratic variable $\lambda \, \widetilde\Omega$ by defining $\Omega = \lambda \, \widetilde\Omega$, which makes the objective $\lambda$ equal to $\tr \Omega/d_O$. Remembering also that the normalisation $1/(1+\lambda)$ is irrelevant for conic constraints, the problem reduces to minimizing $\tr \Omega/d_O$ such that
\begin{equation}
 W + \Omega \in \wsepcone , \quad \Omega \in \mathcal W \, .
\end{equation}

To translate this SDP problem into the language of Definition~\ref{def:sdp_sovietic}, let us define
\begin{gather}
 E = \mathcal L_V \times \mathcal L_V \, , \\
 \mathcal K = \wsepcone \times \mathcal W \, , \\
 \mathcal L = \{ (\Omega, \Omega) \, | \, \Omega \in \mathcal L_V \} \, , \\
 b = (W,0)\, , \\
 c = (0,\id/d_O)\, ,
\end{gather}
and the inner product
\begin{equation}
 \prin{(S,\Sigma)}{(W,\Omega)} = \tr(S W) + \tr(\Sigma \Omega).
\end{equation}
With these definitions, and denoting by $x = (\omega, \Omega)$ its variable, the primal SDP problem becomes
\begin{equation}\label{eq:primal}
\begin{gathered}
 \min\quad \Big\langle (0,\id/d_O) \, , \, (\omega,\Omega) \Big\rangle\\
 \text{s.t.}\quad (\omega,\Omega) \in \big( \wsepcone\! \times\! \mathcal W \big) \cap \{\omega {=} W {+} \Omega\} \, ,
\end{gathered}
\end{equation}
which indeed corresponds to the SDP problem~\eqref{eq:sdp_dual}.

To construct the dual SDP problem, first note that
\begin{gather}
  E' = \mathcal L_V \times \mathcal L_V \, , \\
 \mathcal K^* = \mathcal S_V \times \mathcal W_V^* \, , \label{eq:cone_cartesian} \\ 
 \mathcal L^\perp = \{ (S, -S) \, | \, S \in \mathcal L_V \}
\end{gather}
where we used the property that $(\mathcal K_1 \times \mathcal K_2)^* = \mathcal K_1^* \times \mathcal K_2^*$ in equation \eqref{eq:cone_cartesian}. Denoting by $y = (S, \Sigma)$ its variable, the dual SDP problem is then
  \begin{equation}\label{eq:dual}
\begin{gathered}
 \min\quad \Big\langle (S,\Sigma) \, , \, (W,0) \Big\rangle\\
 \text{s.t.}\quad (S,\Sigma) \in \big( \mathcal S_V \times \mathcal W^*_V \big) \cap \{\Sigma = \id/d_O - S\} \, ,
\end{gathered}
\end{equation}
which corresponds to the SDP problem \eqref{eq:sdp_witness}. 

Let us emphasize that here the duals of $\wsepcone$ and $\mathcal W$ are, respectively, $\mathcal S_V$ and $\mathcal W^*_V$, instead of $\mathcal S$ and $\mathcal W^*$, which is a consequence of choosing the vector space $E$ to be $E = F \times F$ with $F = \mathcal L_V$ instead of $F = {A_I}\otimes {A_O}\otimes{B_I}\otimes{B_O}$. We did this because as subsets of ${A_I}\otimes {A_O}\otimes{B_I}\otimes{B_O}$, the cones $\wsepcone$ and $\mathcal W$ (and therefore $\mathcal K = \wsepcone \times \mathcal W$) have empty interiors, and therefore these cones would not satisfy the requirements of Definition~\ref{def:sdp_sovietic}. This is problematic because the duals of cones with empty interior are not pointed (in our case, $\mathcal S$ and $\mathcal W^*$ are not pointed), and algorithms that solve SDP problems are numerically unstable when optimizing over non-pointed cones.

This definition is indeed satisfied by the cones we chose, \ie, $\wsepcone \times \mathcal W \subseteq \mathcal L_V \times \mathcal L_V$ is indeed pointed and has nonempty interior, as we shall check now. A pointed cone $\mathcal K$ is a cone such that $\mathcal K \cap (-\mathcal K) = \{0\}$. This indeed satisfied for $\wsepcone \times \mathcal W$, as both cones require their elements to be positive semidefinite, and $W \ge 0$ and $-W \ge 0$ imply that $W = 0$. To show that $\wsepcone \times \mathcal W$ has nonempty interior, it is enough\footnote{Since $\operatorname{int} (\wsepcone \times \mathcal W) = \operatorname{int} \wsepcone \times \operatorname{int} \mathcal W$ and $\wsepcone \subseteq \mathcal W$.} to find an operator that belongs to $\operatorname{int} \wsepcone$. This is done through the following lemma:
\begin{lemma}\label{lemma:wsepcone_interior}
 $\id^\circ + \Omega \in \operatorname{int} \wsepcone$ for any $\Omega \in \mathcal L_V$ such that $\norm{\Omega}_2 < \frac 1{2d_I}$, where $d_I = d_{A_I}d_{B_I}$ and $\norm{\cdot}_2$ is the Hilbert-Schmidt norm.
\end{lemma}
\begin{proof}
Since $\Omega \in \mathcal L_V$, the discussion in section \ref{sec:causally_separable_2} implies that the operators
 \begin{gather}
 \omega^{A \prec B} := {}_{B_O}\Omega  \, , \quad
 \omega^{B \prec A} := \Omega - {}_{B_O}\Omega
\end{gather}
are causally ordered (in the sense that they satisfy Eq.~\eqref{eq:valid_W_AB} and the analogous relation for the order $B \prec A$, respectively), and so are the operators
\begin{gather}
 W^{A \prec B} := \frac12 \id^\circ + \omega^{A \prec B} \\
 W^{B \prec A} := \frac12 \id^\circ + \omega^{B \prec A}
\end{gather}
Since, furthermore, 
\begin{equation}
 W^{A \prec B} + W^{B \prec A} = \id^\circ + \Omega,
\end{equation}
we have that $\id^\circ + \Omega \in \wsepcone$ if $W^{A \prec B}$ and $W^{B \prec A}$ are positive semidefinite. This is the case if
\begin{equation}
\norm{\omega^{A \prec B}} \le \frac 1{2d_I} \mathand \norm{\omega^{B \prec A}} \le \frac 1{2d_I} \, ,
\end{equation}
where $\norm{\cdot}$ is the standard operator norm (\ie, the maximum singular value). To be able to enforce that, first note that $\omega^{A \prec B}$ and $\omega^{B \prec A}$ are orthogonal, and therefore Pythagoras' theorem implies that
 \begin{equation}
 \norm{\Omega}_2^2 = \norm{\omega^{A \prec B}}_2^2 + \norm{\omega^{B \prec A}}_2^2 \, ,
\end{equation}
which implies that
\begin{multline}\label{eq:norm_inequality}
\max \DE{ \norm{\omega^{A \prec B}}, \norm{\omega^{B \prec A}}} \le \\ \max \DE{ \norm{\omega^{A \prec B}}_2, \norm{\omega^{B \prec A}}_2} \le \norm{\Omega}_2 \, ,
\end{multline}
and therefore
\begin{equation}\label{eq:closed_ball}
 \norm{\Omega}_2 \le \frac 1{2d_I}
\end{equation}
implies that $\id^\circ + \Omega \in \wsepcone$. This in turn implies that the interior of the ball composed of operators $\id^\circ + \Omega$ with $\Omega$ satisfying \eqref{eq:closed_ball} belongs to the interior of $\wsepcone$, \ie, $\norm{\Omega}_2 < \frac 1{2d_I}$ implies that $\id^\circ + \Omega \in \operatorname{int} \wsepcone$.
\end{proof}
  
This concludes the proof that problems \eqref{eq:sdp_witness} and \eqref{eq:sdp_dual} are SDP problems dual to each other. We shall now proceed to show that the Duality Theorem (Theorem~4.2.1 in~\cite{nesterov87}) applies to them:
\begin{theorem}\label{thm:duality}
 Let (P), (D) be a primal-dual pair of conic problems as defined above, and let the pair be such that
 \begin{enumerate}[leftmargin=5mm]
  \item The set of primal solutions $\mathcal K \cap (\mathcal L+b)$ intersects $\operatorname{int} \mathcal K$;
  \item The set of dual solutions $\mathcal K^* \cap (\mathcal L^\perp+c)$ intersects $\operatorname{int} \mathcal K^*$;
  \item $\prin{c}{x}$ is lower bounded for all $x \in \mathcal K \cap (\mathcal L + b)$.
 \end{enumerate}
Then both the primal and the dual problems are solvable, and the optimal solutions $x^*$ and $y^*$ satisfy the relation
\begin{equation}\label{eq:duality}
 \prin{c}{b} = \prin{c}{x^*} + \prin{y^*}{b}.
\end{equation}
\end{theorem}

Let us check that for the SDP problems \eqref{eq:primal} and \eqref{eq:dual}, the three assumptions of the Duality Theorem are indeed satisfied.

To see that \emph{1.} is satisfied, we need to find $\Omega \in \operatorname{int} \mathcal W$ such that $W + \Omega \in \operatorname{int} \wsepcone$. Take $\Omega = \lambda \id^\circ$; then $W + \lambda\id^\circ \in \operatorname{int} \wsepcone$ iff $\frac1\lambda W + \id^\circ \in \operatorname{int} \wsepcone$. Using Lemma~\ref{lemma:wsepcone_interior}, we conclude that this is true if
\begin{equation}\label{eq:it_belongs_to_the_interior}
 \norm{\frac1\lambda W}_2 < \frac1{2d_I}
\end{equation}
Since
\begin{equation}
 \norm{\frac1\lambda W}_2 \le \frac1\lambda \norm{W}_1 = \frac{d_O}{\lambda}
\end{equation}
it is enough to choose
\begin{equation}
 \lambda > 2d_Id_O
\end{equation}
to satisfy inequality \eqref{eq:it_belongs_to_the_interior}, and we're done.

To see that \emph{2.} is satisfied, we need to exhibit a witness $S$ such that $S \in \operatorname{int}\mathcal S_V$ and $\id/d_O - S \in \operatorname{int}\mathcal W^*_V$. Since the cone $\mathcal P \cap \mathcal L_V$ of positive semidefinite matrices in $\mathcal L_V$ is a full-dimensional subset of both $\mathcal S_V$ and $\mathcal W^*_V$, it is enough to find an operator $S$ such that $S > 0$ and $\id/d_O -S > 0$. One can take $S = \id/(2d_O)$.

To see that \emph{3.} is satisfied, note that $\Omega \ge 0$ implies that $\tr\Omega/d_O \ge 0$.

All in all, the three assumptions of the Duality Theorem above are thus satisfied. Applying the identity~\eqref{eq:duality} to our pair of conic problems, we have, for the optimal solutions $\Omega^*$ and $S^*$:
\begin{equation}
 0 = \tr[\Omega^*]/d_O + \tr[S^*W] \, ,
\end{equation}
as claimed in Eq.~\eqref{eq:duality_relation}.
As discussed in Sec.~\ref{sec:testing_separability}, a value $\tr[\Omega^*]/d_O  = -\tr[S^* W] > 0$ guarantees that the process matrix $W$ is causally nonseparable, and the solution $S^*$ of the dual problem provides an explicit causal witness; a value $\tr[\Omega^*]/d_O  = 0$ proves that the process matrix $W$ is causally separable, and the primal problem provides a decomposition of $W$ in terms of causally ordered process matrices $W^{A \prec B}$ and $W^{B \prec A}$ (again, this is easier to see in the representation of the primal problem shown in \eqref{eq:sdp_dual_D2}).

\section{Measuring causal nonseparability}\label{sec:measures}

  A causal witness can be used not only to detect the causal nonseparability of a given process, but also to \textit{measure} it. This is analogous to the situation with entanglement witnesses and entanglement measures \cite{brandao05}. First of all, we need to define what we mean by a measure of causal nonseparability. In analogy with the case of entanglement, we suggest that a proper measure of causal nonseparability $\mathcal N$ should satisfy the following properties:
\begin{description}
 \item[Discrimination] $\mathcal N(W) \ge 0$ for every process matrix $W$, with $\mathcal N(W) = 0$ if and only if $W$ is causally separable.
 \item[Convexity] $\mathcal N\de{\sum_ip_iW_i} \le \sum_ip_i\mathcal N(W_i)$ for any process matrices $W_i$ and any $p_i \ge 0$, with $\sum_i p_i = 1$.
 \item[Monotony] $\mathcal N\big(\$(W)\big) \le \mathcal N(W)$, where $\$(W)$ is any process obtainable from $W$ by composing it with local CPTP maps. 
\end{description}
  
Now we shall prove that both $\gr(W)$ and $\rr(W)$ as defined in equations \eqref{eq:def_generalised_robusntess} and \eqref{eq:def_random_robusntess} respect the properties of \textbf{Discrimination} and \textbf{Convexity}, whereas $\gr(W)$ respects \textbf{Monotony} but $\rr(W)$ does not.

\textbf{Discrimination} follows from the definition of the SDP problems~\eqref{eq:sdp_witness}--\eqref{eq:sdp_dual} and \eqref{eq:sdp_random_robustness_primal}--\eqref{eq:sdp_random_robustness}. Note that since they satisfy the assumptions of the Duality Theorem (\ref{thm:duality}), there are algorithms that actually find the optimal solutions efficiently.

To demonstrate \textbf{Convexity}, let us denote by $S_W$ the optimal witness for a given process matrix $W$; because of its optimality, one has, for any process matrices $W_i$ and any $p_i \ge 0$,
\begin{equation}
 \tr[S_{W_j} W_j] \le \tr\De{\de{S_{\sum_ip_iW_i}} W_j}
\end{equation}
and therefore
\begin{equation}
 - \tr\Big[\de{S_{\sum_ip_iW_i}} \sum_ip_iW_i\Big] \le -\sum_ip_i\tr[S_{W_i} W_i],
\end{equation}
that is, $\mathcal N\de{\sum_ip_iW_i} \le \sum_ip_i\mathcal N(W_i)$.

Now we show that \textbf{Monotony} does hold for $\gr(W)$. For that, first we need to define the map $\$(\cdot)$ that composes a process $W$ with local operations. More specifically, the map $\$(\cdot)$ composes a process with the CPTP map $M_1^A$ applied to Alice's input, the CPTP map $M_3^A$ applied to Alice's output, the CPTP map $M_1^B$ applied to Bob's input, and the CPTP map $M_3^B$ applied to Bob's output. We can then define $\$(\cdot)$ as the map such that for all processes $W$ and all CP maps $C^A_2$ and $C^B_2$ we have that
\begin{equation}
 \tr[ (C^A_2 \otimes C^B_2) \cdot \$(W)] = \tr[ (C_{123}^A \otimes C_{123}^B) W],
\end{equation}
where
\begin{equation}
 C_{123}^X := \De{\mathcal I \otimes (\mathcal M^X_3 \circ \mathcal C^X_2 \circ \mathcal M^X_1)(\KetBra{\id})}^T
\end{equation}
is the Choi-Jamiołkowski operator of the composition of the each party's operations. The processes $W$ and $\$(W)$ are illustrated in Figure~\ref{fig:local-maps}. 

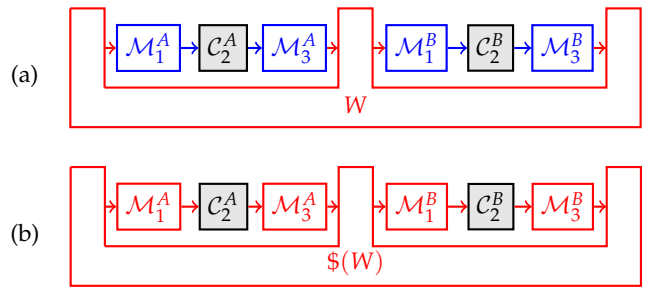
\begin{figure}[htpc]
  \centering
  \begin{tikzpicture}[scale=1.5]
		\node[draw, thick, rectangle,minimum width=0.4cm,minimum height=0.4cm,fill=white!80!gray] (A) at (2.5/2+0.18/2,-0.35) {$\mathcal{C}^{A}_{2}$};
		\node[draw, thick, rectangle,minimum width=0.4cm,minimum height=0.4cm,fill=white!80!gray] (B) at (3*2.5/2-0.13/2,-0.35) {$\mathcal C^B_2$};
              	\node[draw, thick, rectangle,minimum width=0.4cm,minimum height=0.4cm, blue] (MA1) at (2.5/2+0.15/2-0.64,-0.35) {$\mathcal{M}^{A}_1$};
		\node[draw, thick, rectangle,minimum width=0.4cm,minimum height=0.4cm, blue] (MA3) at (2.5/2+0.15/2+0.64,-0.35) {$\mathcal M^A_3$};
                \node[draw, thick, rectangle,minimum width=0.4cm,minimum height=0.4cm, blue] (MB1) at (3*2.5/2-0.15/2-0.64,-0.35) {$\mathcal{M}^{B}_1$};
		\node[draw, thick, rectangle,minimum width=0.4cm,minimum height=0.4cm, blue] (MB3) at (3*2.5/2-0.15/2+0.64,-0.35) {$\mathcal M^B_3$};
    \draw[thick, blue,->] (A) -- (MA3);  \draw[thick, blue,<-]  (A) -- (MA1);
    \draw[thick, blue,->] (B) -- (MB3);  \draw[thick, blue,<-]  (B) -- (MB1);
    \draw[thick, red,->] (MA3) -- ++(0.38,0);  \draw[thick, red,<-]  (MA1) --  ++(-0.38,0);
    \draw[thick, red,->] (MB3) -- ++(0.38,0);  \draw[thick, red,<-]  (MB1) --  ++(-0.38,0);
    \draw[thick, red] (0,0) -- ++(0,-1.05) -- ++(5,0) -- ++ (0,1.05) -- ++(-0.3,0) -- ++(0,-0.7) -- (2.5+0.15,-1+0.3) -- ++(0,0.7) -- ++(-0.3,0) -- ++(0,-0.7) -- (0.3,-0.7);
    \draw[thick, red] (0,0) --  ++(+0.3,0) -- ++(0,-0.7);
    \node[thick, red] at (2.5,-0.85) {$W$};
    \node[thick] at (-0.4,-0.6) {(a)};
    \node[] at (0,-1.3) {};
  \end{tikzpicture}
  \begin{tikzpicture}[scale=1.5]
		\node[draw, thick, rectangle,minimum width=0.4cm,minimum height=0.4cm,fill=white!80!gray] (A) at (2.5/2+0.18/2,-0.35) {$\mathcal{C}^{A}_{2}$};
		\node[draw, thick, rectangle,minimum width=0.4cm,minimum height=0.4cm,fill=white!80!gray] (B) at (3*2.5/2-0.13/2,-0.35) {$\mathcal C^B_2$};
              	\node[draw, thick, rectangle,minimum width=0.4cm,minimum height=0.4cm, red] (MA1) at (2.5/2+0.15/2-0.64,-0.35) {$\mathcal{M}^{A}_1$};
		\node[draw, thick, rectangle,minimum width=0.4cm,minimum height=0.4cm, red] (MA3) at (2.5/2+0.15/2+0.64,-0.35) {$\mathcal M^A_3$};
                \node[draw, thick, rectangle,minimum width=0.4cm,minimum height=0.4cm, red] (MB1) at (3*2.5/2-0.15/2-0.64,-0.35) {$\mathcal{M}^{B}_1$};
		\node[draw, thick, rectangle,minimum width=0.4cm,minimum height=0.4cm, red] (MB3) at (3*2.5/2-0.15/2+0.64,-0.35) {$\mathcal M^B_3$};
    \draw[thick, red,->] (A) -- (MA3);  \draw[thick, red,<-]  (A) -- (MA1);
    \draw[thick, red,->] (B) -- (MB3);  \draw[thick, red,<-]  (B) -- (MB1);
    \draw[thick, red,->] (MA3) -- ++(0.38,0);  \draw[thick, red,<-]  (MA1) --  ++(-0.38,0);
    \draw[thick, red,->] (MB3) -- ++(0.38,0);  \draw[thick, red,<-]  (MB1) --  ++(-0.38,0);
    \draw[thick, red] (0,0) -- ++(0,-1.05) -- ++(5,0) -- ++ (0,1.05) -- ++(-0.3,0) -- ++(0,-0.7) -- (2.5+0.15,-1+0.3) -- ++(0,0.7) -- ++(-0.3,0) -- ++(0,-0.7) -- (0.3,-0.7);
    \draw[thick, red] (0,0) --  ++(+0.3,0) -- ++(0,-0.7);
    \node[thick, red] at (2.5,-0.85) {$\$(W)$};
    \node[thick] at (-0.4,-0.6) {(b)};
  \end{tikzpicture}  \caption{(a) The situation where the parties share a bipartite process $W$ (in red) and apply the CPTP maps $\mathcal M^X_1$ and $\mathcal M^X_3$ (in blue) to their inputs and outputs can be equivalently described by (b) a single bipartite process $\$(W)$ (in red).}
  \label{fig:local-maps}
\end{figure}

It follows from this definition that $\$(W)$ is a valid process. To see this, note that the validity of $W$ implies that the probabilities
\begin{equation}
 P(\mathcal C_{123}^A,\mathcal C_{123}^B) = \tr[ (C_{123}^A \otimes C_{123}^B) W]
\end{equation}
are positive and normalised. By definition, these are equal to the probabilities
\begin{equation}
\label{eq:probability-composition}
P(\mathcal C^A_2,\mathcal C^B_2) = \tr[ (C^A_2 \otimes C^B_2) \cdot \$(W)],
\end{equation}
and the arguments in Appendix~\ref{sec:valid_w} show that requiring the probabilities $P(\mathcal C^A_2,\mathcal C^B_2)$ to be positive and normalised is enough to imply the validity of the process $\$(W)$.

Furthermore, if $W$ is causally separable so is $\$(W)$. This follows from the linearity of $\$(\cdot)$ and from the fact that $\$(\cdot)$ preserves the causal order when applied to a causally ordered process, which follows directly from the analogous property for quantum combs \cite{chiribella09b}. 

We want to show that for all $\$(\cdot)$ (\ie, for all CPTP maps $M_1^A$, $M_3^A$, $M_1^B$ and $M_3^B$) and $W$,
\begin{equation}
 \gr(\$(W)) \le \gr(W),
\end{equation}
or equivalently that
\begin{equation}
 -\tr\De{\de{S_{\$(W)}}\$(W)} \le -\tr[(S_W) W] \, .
\end{equation}
By duality, this is equivalent to
\begin{equation}
 -\tr\De{\$^*\de{S_{\$(W)}} W} \le -\tr[(S_W) W]
\end{equation}
(where $\$^*$ is the dual map of $\$$), which follows from the optimality of $S_W$ if $\$^*\de{S_{\$(W)}}$ is a valid causal witness that respects the normalisation condition for generalised robustness (as defined in SDP problem~\eqref{eq:sdp_witness}). Therefore, we need to show it has the two following properties:
\begin{gather}
 \tr\De{\$^*\de{S_{\$(W)}} W^\text{sep}} \ge 0 \quad \forall W^\text{sep} \, , \\
 \id/d_O - \$^*\de{S_{\$(W)}} \in \mathcal W^* \, .
\end{gather}
The first one follows from duality
\begin{equation}
 \tr\De{\$^*\de{S_{\$(W)}} W^\text{sep}} = \tr\De{S_{\$(W)} \$\de{W^\text{sep}}}
\end{equation}
and the fact that $\$\de{W^\text{sep}}$ is causally separable and $S_{\$(W)}$ is a causal witness.

The second one is equivalent to
\begin{equation}
 \tr\De{\de{\id/d_O - \$^*\de{S_{\$(W)}}}\Omega} \ge 0
\end{equation}
for every (not necessarily normalised) process matrix $\Omega$. From duality and linearity this is equivalent to
\begin{equation}
 \tr\De{ S_{\$(W)} \, \$(\Omega)} \le \tr(\Omega)/d_O \, ,
\end{equation}
and this follows from the fact that $\$(\cdot)$ is trace-preserving and that $\id/d_O - S_{\$(W)} \in \mathcal W^*$ (which is the normalization condition from the SDP problem~\eqref{eq:sdp_witness}).

An analogous proof fails for random robustness, as the dual map $\$^*(\cdot)$ can increase the trace of a witness, and therefore make it fail to satisfy the normalisation condition for SDP problem~\eqref{eq:sdp_random_robustness}. To show that $\rr(W)$ does not in fact satisfy \textbf{Monotony}, it is enough to find a process and local operations such that $\rr\big(\$(W)\big) > \rr(W)$. 

A concrete counterexample can be obtained by considering $W_\text{OCB}$ and $S_\text{OCB}$ from section \ref{sec:witness_example}. Let
\begin{equation}
 W_1 = W_\text{OCB} \otimes \frac{\id^{A_I'}}{2}
\end{equation}
be the process obtained from $W_\text{OCB}$ by adding a maximally mixed qubit to Alice's input space. Then its random robustness is (up to numerical precision)
\begin{equation}
 \rr(W_1) = -\tr S_{W_1}\,W_1 = \sqrt{2}-1,
\end{equation}
where
\begin{equation}
 S_{W_1} = 2\, S_\text{OCB} \otimes \proj{0}^{A_I'}
\end{equation}
is its optimal random robustness witness. Now, we can obtain the process
\begin{equation}
 \$(W_1) = W_\text{OCB} \otimes \proj{0}^{A_I'}
\end{equation}
from $W_1$ simply by discarding the system in Alice's input space $A_I'$ and replacing it with $\proj{0}$, which is clearly a local operation. Then its random robustness is (up to numerical precision)
\begin{equation}
 \rr(\$(W_1)) = -\tr S_{\$(W_1)}\,\$(W_1) =  2(\sqrt{2}-1),
\end{equation}
where $S_{\$(W_1)} = S_{W_1}$. Thus we have shown that
 \begin{equation}
  \rr(\$(W_1)) > \rr(W_1),
 \end{equation}
so random robustness is not monotonous under local operations.

\section{Characterisation of tripartite causal witnesses}\label{app:proof_thm3}

\begin{proof}[Proof of Theorem~\ref{thm:switch_witness}]
 As defined in section \ref{sec:causally_separable_3c}, the cone of tripartite causally separable processes with ${d_{C_O}=1}$ is
 \begin{equation}
  \wsepcone_{3C} = \conv\De{ \de{\mathcal P \cap \mathcal L_{A \prec B \prec C}} \cup \de{\mathcal P \cap \mathcal L_{B \prec A \prec C}} } \, ,
 \end{equation}
where $\mathcal L_{A \prec B \prec C}$ and $\mathcal L_{B \prec A \prec C}$ are the linear subspaces defined by the projectors $L_{A \prec B \prec C}$ and $L_{B \prec A \prec C}$. The cone of causal witnesses $\mathcal S_{3C}$ is its dual
\begin{equation}
 \mathcal S_{3C} = {\wsepcone_{3C}}^* \, .
\end{equation}
Using duality relations \eqref{eq:dual_properties} and \eqref{eq:orthogonal_complement} and the fact that the cone of positive semidefinite matrices is self-dual, we get
  \begin{align}
  \mathcal S_{3C} &=  \de{\mathcal P \cap \mathcal L_{A \prec B \prec C}}^* \cap \de{\mathcal P \cap \mathcal L_{B \prec A \prec C}}^* \\
  &=  \De{\conv\de{\mathcal P \cup \mathcal L_{A \prec B \prec C}^\perp}} \cap \De{\conv\de{\mathcal P \cup \mathcal L_{B \prec A \prec C}^\perp}},
 \end{align}
with
\begin{multline}
 \conv\de{\mathcal P \cup \mathcal L_{A \prec B \prec C}^\perp} = \\ \{ S^P_{ABC} + S^\perp_{ABC} \midset S^P_{ABC} \ge 0,\ L_{A \prec B \prec C}(S^\perp_{ABC})=0 \}
\end{multline}
and
\begin{multline}
 \conv\de{\mathcal P \cup \mathcal L_{B \prec A \prec C}^\perp} = \\ \{ S^P_{BAC} + S^\perp_{BAC} \midset S^P_{BAC} \ge 0,\ L_{B \prec A \prec C}(S^\perp_{BAC})=0 \} \, .
\end{multline}
\end{proof}
\section{Optimizing Chiribella's task}\label{sec:optimizing_chiribella}

We want to optimize the weights $q_{ij}^\com, q_{ij}^\ant$ so as to minimize the maximal probability of success\footnote{Remember that the probability of success for the quantum switch is always equal to one.} for causally separable processes $p_\text{succ}^\text{sep}$, \ie, we want to minimize the upper bound
\begin{equation}
 \tr(G_\text{finite} W^\text{sep}) \le p_\text{succ}^\text{sep}.
\end{equation}
This is relevant because, according to equation \eqref{eq:generalised_robustness_versus_p_succ}, a lower $p_\text{succ}^\text{sep}$ corresponds to a larger resistance to worst-case noise.

To do this, note that $\tr(G_\text{finite} W) \le p_\text{succ}^\text{sep}$ if and only if $\tr[(p_\text{succ}^\text{sep}\id/d_O-G_\text{finite}) W] \ge 0$. Imposing that this holds for all causally separable processes $W \in \wsepcone_{3C}$ amounts to imposing that $p_\text{succ}^\text{sep}\frac{\id}{d_O} - G_\text{finite} \in \mathcal S_{3C}$, where $\mathcal S_{3C}$ is the cone of causal witnesses (characterized through Theorem~\ref{thm:switch_witness}).

We are thus led to define the following SDP problem:
\begin{equation}
\begin{gathered} \min \ p_\text{succ}^\text{sep} \\[1mm]
 \text{s.t.} \qquad p_\text{succ}^\text{sep}\id/d_O - G_\text{finite} \in \mathcal S_{3C} \, , \quad \\
  q_{ij}^\com \ge 0 \, , \quad q_{ij}^\ant \ge 0 \, , \quad \sum_{i,j=1}^{10} q_{ij}^\com  + q_{ij}^\ant = 1 \, ,\\
  q_{ij}^\com = 0\quad\quad \forall i,j\quad\text{s.t.}\quad[U_i,U_j]\neq0, \\ 
  q_{ij}^\ant = 0\quad\quad \forall i,j\quad\text{s.t.}\quad\{U_i,U_j\}\neq0,
\end{gathered} \label{SDP:optim_chiribella}
\end{equation}
where in order to keep the interpretation of the task as guessing whether the unitaries commute or anticommute, we imposed that $q_{ij}^\com = 0$ for non-commuting $U_i, U_j$ and $q_{ij}^\ant = 0$ for non-anticommuting $U_i,U_j$.

Solving this problem numerically, we found
\begin{gather}
p_\text{succ}^\text{sep} \approx 0.8690
\end{gather}
(and we omit the optimal $q_{ij}^\com, q_{ij}^\ant$ for brevity).


\clearpage

\bibliographystyle{linksen}
\bibliography{biblio}

\end{document}